\documentclass[letterpaper, 10 pt, journal, twoside]{IEEEtran}

\IEEEoverridecommandlockouts                              %

\usepackage{cite}
\usepackage{amsmath,amssymb,amsfonts,amsthm}
\usepackage{bm, bbm, hyperref}
\usepackage{nicefrac}
\usepackage{xcolor}
\usepackage[ruled,vlined]{algorithm2e}
\usepackage{graphicx}

\newcommand{\N}[0]{\ensuremath{\mathbb{N}}}
\newcommand{\R}[0]{\ensuremath{\mathbb{R}}}
\newcommand{\Id}[0]{\ensuremath{\operatorname{Id}}}
\newcommand{\norm}[1]{\left\lVert#1\right\rVert}

\DeclareMathOperator*{\st}{subject~to}
\DeclareMathOperator*{\argmin}{arg\,min}

\newtheorem{assumption}{Assumption}

\newtheorem{lemma}{Lemma}
\newtheorem{proposition}{Proposition}
\newtheorem{remark}{Remark}
\newtheorem{theorem}{Theorem}
\newtheorem{example}{Example}

\usepackage[framemethod=tikz]{mdframed}

\usepackage{enumitem}

\usepackage[acronym]{glossaries}
\usepackage[nohyperlinks]{acronym}

\newacronym{iid}{i.i.d.}{independent and identically distributed}
\newacronym{wrt}{w.r.t.}{with respect to}
\newacronym{wlog}{w.l.o.g.}{without loss of generality}
\newacronym{PAC}{PAC}{probably approximately correct}
\newacronym{SNEP}{SNEP}{stochastic Nash equilibrium problem}
\newacronym{SNE}{SNE}{stochastic Nash equilibrium}
\newacronym{SGD}{SGD}{stochastic gradient descent}
\newacronym{NE}{NE}{Nash equilibrium}
\newacronym{GNE}{GNE}{generalized Nash equilibrium}
\newacronym{SGNEP}{SGNEP}{stochastic generalized Nash equilibrium problem}
\newacronym{KKT}{KKT}{Karush–Kuhn–Tucker}
\newacronym{FB}{FB}{forward-backward}
\newacronym{FBF}{FBF}{forward-backward-forward}
\newacronym{FPI}{FPI}{fixed-point iteration}
\newacronym{EG}{EG}{extragradient}
\newglossaryentry{VI}
{
	name={VI},
	description={variational inequality},
	first={\glsentrydesc{VI} (\glsentrytext{VI})},
	plural={VIs},
	descriptionplural={variational inequalities},
	firstplural={\glsentrydescplural{VI} (VIs)}
}
\newacronym{DR}{DR}{Douglas-Rachford}
\newacronym{DY}{DY}{Davis-Yin}
\newacronym{SAA}{SAA}{sample average approximation}
\newacronym{SA}{SA}{stochastic approximation}
\newacronym{PEV}{PEV}{plug-in electric vehicle}
\newacronym{DSO}{DSO}{distribution system operator}

\title{\LARGE \bf
Finite-sample guarantees for data-driven operator splitting methods via martingale inequalities
}

\author{Andrea Martin, Filippo Fabiani, and Giuseppe Belgioioso
\thanks{A. Martin and G. Belgioioso are with the School of Electrical Engineering and Computer Science, and Digital Futures, KTH Royal Institute of Technology, Sweden. E-mail addresses: \{andrmar, giubel\}@kth.se.}
\thanks{Filippo Fabiani is with the IMT School for Advanced Studies Lucca,
Italy. E-mail address: filippo.fabiani@imtlucca.it.}
\thanks{This work was supported by Digital Futures and the Wallenberg AI, Autonomous Systems and Software Program (WASP) funded by the Knut and Alice Wallenberg Foundation.}
}

\begin{document}

\maketitle

\thispagestyle{empty}
\pagestyle{empty}

\begin{abstract}
  Operator splitting methods are a fundamental class of algorithms for solving structured monotone inclusion problems arising in optimization, control, and game theory. We consider the case, common in stochastic regimes, where the forward evaluation of one of the constituent operators is either unavailable in closed form or computationally expensive to evaluate, and is therefore approximated using a finite number of noisy oracle samples. We establish distribution-free, finite-sample certificates for the quality of the output produced by data-driven \gls*{DY} splitting algorithms. Unlike previous works, our analysis directly controls the residual error via martingale inequalities instead of relying on algorithmic stability arguments for a tailored surrogate loss, yielding the first a priori certificates whose statistical excess provably vanishes with the sample size. We further show that, under linear convergence of the \gls*{DY} splitting algorithm, the dependence of our bounds on the iteration count improves from linear growth to exponential decay. We validate our theoretical results on a stochastic portfolio optimization problem with uncertain asset returns.
\end{abstract}

\glsresetall

\section{Introduction}
Monotone inclusions provide a unifying mathematical modeling framework for a broad class of decision-making problems arising in optimization, control, and game theory. In fact, the solution sets of many convex optimization problems, variational inequalities, and saddle-point or equilibrium problems admit equivalent characterizations as zeros of the sum of two or more monotone operators~\cite{bauschke2017convex, ryu2022large}. Operator splitting schemes are a fundamental class of algorithms for computing such zeros. These methods exploit the structure of the individual operators to replace the direct solution of the original inclusion with a sequence of simpler subproblems, often resulting in low-complexity iterations that are well-suited to embedded control~\cite{stathopoulos2016operator} and signal processing~\cite{combettes2011proximal} applications.

In stochastic settings, the evaluations of the individual operators frequently involve expectations \gls*{wrt} an underlying probability distribution. In many applications, such as resource allocation problems with uncertain demand, portfolio optimization with uncertain asset returns~\cite{markowitz1955portfolio, boyd2017multi}, or regression problems with random input-output data~\cite{bousquet2002stability}, this distribution is unknown and accessible only through samples, preventing the required operator evaluations from being computed exactly.

To address this challenge, two main data-driven schemes for approximating the required operator evaluations have been considered in the literature: \gls*{SAA}~\cite{shapiro2003monte, shapiro2021lectures} and \gls*{SA}~\cite{robbins1951stochastic}. In \gls*{SAA}, the unavailable expectation is replaced by an empirical average over a finite set of \gls*{iid} samples, resulting in a deterministic approximation of the original problem~\cite{franci2020distributed}. In contrast, \gls*{SA} methods rely on a series of stochastic estimates constructed by drawing new samples at each iteration of the algorithm. This substantially reduces the computational cost of each iteration, at the price of introducing sampling noise and typically requiring suitable step-size selection rules and additional assumptions to guarantee convergence~\cite{jiang2008stochastic, koshal2012regularized, yousefian2017smoothing}. To mitigate this tradeoff, several variance-reduced \gls*{SA} strategies have also been proposed. These include methods that progressively increase the number of samples used at each iteration~\cite{iusem2017extragradient, franci2021stochastic}, as well as schemes that combine inexpensive stochastic updates with occasional computations over larger batches~\cite{gower2020variance, alacaoglu2022stochastic}. Besides their different constructions, this body of literature \cite{franci2020distributed, jiang2008stochastic, koshal2012regularized, yousefian2017smoothing, iusem2017extragradient, franci2021stochastic, gower2020variance, alacaoglu2022stochastic} mainly focuses on asymptotic convergence to a solution of the original inclusion as the total number of samples used by an iterative algorithm grows unbounded, and does not characterize the accuracy that can be reached when only a limited number of data points is available and generating new ones is costly.

Motivated by this limitation, several works have recently investigated the finite-sample behavior of data-driven iterative algorithms based on sample average or stochastic approximations. To address this question, a prominent line of research, pioneered for stochastic optimization by~\cite{hardt2016train}, builds upon the algorithmic stability framework of~\cite{bousquet2002stability}, which relates the sensitivity of a learning algorithm to perturbations in its input dataset to its generalization performance. In particular, stability properties of stochastic gradient methods have received considerable attention, with early contributions establishing stability bounds in both convex and nonconvex settings~\cite{hardt2016train,bassily2020stability}. This line of analysis was subsequently extended to minimax problems~\cite{lei2021stability}, adversarial learning~\cite{xing2021algorithmic}, and variational inequalities~\cite{zhao2024learning}. More recently, this perspective has also been brought to operator splitting methods, with~\cite{fabiani2026finite} deriving finite-sample guarantees for data-driven \gls*{FB} schemes via algorithmic stability of a scalar surrogate loss for the inclusion residual. 

In this paper, we argue that, while algorithmic stability provides an elegant framework for studying generalization in expected-risk problems~\cite{hardt2016train, bassily2020stability, lei2021stability, xing2021algorithmic}, its application to residual certification in monotone inclusion problems through surrogate losses~\cite{fabiani2026finite, fabiani2026certifying} may instead lead to bounds whose statistical excess need not vanish with the sample size. This observation raises the question of whether one can obtain asymptotically consistent finite-sample residual certificates. Focusing on the \gls*{DY} three-operator splitting method~\cite{davis2017three}, a well-known scheme that encompasses the \gls*{FB} splitting as a special case, we answer in the affirmative. Specifically, we first revisit the algorithmic stability argument of~\cite{fabiani2026finite} in the context of \gls*{DY} splitting, obtaining an a posteriori, finite-sample certificate for the true fixed-point residual and, crucially, identifying the use of surrogate losses as a key source of potential conservatism in the analysis. %
To address this limitation, we then observe that the conditional expectations of the statistical error term arising from the \gls*{SAA}, as the samples are revealed one at a time, form a Doob martingale~\cite{doob1953stochastic}. Based on this insight, we establish asymptotically consistent a priori, finite-sample residual certificates by combining a novel replace-one sensitivity bound for the output of the data-driven \gls*{DY} algorithm with concentration inequalities for martingales~\cite{boucheron2013concentration, pinelis1994optimum}. We further show that, under standard assumptions guaranteeing linear convergence of the \gls*{DY} method~\cite{davis2017three}, the dependence of our replace-one sensitivity bound on the number of algorithmic iterations improves from linear growth to exponential decay, leading to sharper finite-sample residual certificates. Last, we present numerical experiments on a portfolio optimization problem with uncertain asset returns to validate our theoretical results.\looseness-1

\smallskip 

\emph{Notation:} We mostly adopt standard operator theoretic notation and definitions from~\cite{bauschke2017convex}. For a closed convex set $\mathcal{Y} \subseteq \R^n$, we denote the projection of $x \in \R^n$ onto $\mathcal{Y}$ by $\operatorname{proj}_{\mathcal{Y}}(x) = \argmin_{y \in \mathcal{Y}}~\|y-x\|$ and the distance of $x$ from $\mathcal{Y}$ by $\operatorname{dist}(x, \mathcal{Y}) = \min_{y \in \mathcal{Y}}~\|y-x\|$. The normal cone of $\mathcal{Y}$ is the operator $\operatorname{N}_{\mathcal{Y}} : \R^n \rightrightarrows \R^n$ defined by $\operatorname{N}_{\mathcal{Y}}(x) = \{v \in \R^n : \sup_{y \in \mathcal{Y}} ~ \langle v, y-x \rangle \leq 0\}$ if $x \in \mathcal{Y}$, and $\operatorname{N}_{\mathcal{Y}}(x) = \emptyset$ otherwise. We denote the graph of an operator $A: \R^n \rightrightarrows \R^n$ by $\operatorname{gra}(A) = \{(x, y) \in \R^n \times \R^n : y \in A(x)\}$. We denote the set of zeros of $A$ by $\operatorname{zer}(A) = \{x \in \R^n : 0 \in A(x)\}$. We say that $A$ is monotone if $\langle x-y, u-v\rangle \geq 0$ for any $(x, u) \in \operatorname{gra}(A)$ and $(y, v) \in \operatorname{gra}(A)$; additionally, $A$ is maximal monotone if there is no other monotone operator $B: \R^n \rightrightarrows \R^n$ such that $\operatorname{gra}(A) \subset \operatorname{gra}(B)$. For any $\theta > 0$, a single-valued operator $A : \R^n \to \R^n$ is $\theta$-cocoercive if $\langle x-y, T(x)-T(y)\rangle \geq \theta \norm{T(x) - T(y)}^2$ for any $x, y \in \R^n$. For any $\gamma > 0$, the resolvent of $A$ is $J_{\gamma A} = (\Id + \gamma A)^{-1}$ and the reflected resolvent of $A$ is $R_{\gamma A} = 2 J_{\gamma A} - \Id$.

\section{Problem formulation}
\label{sec:problem_formulation}

We consider monotone inclusion problems in the form:
\begin{equation}
\label{eq:monotone_inclusion}
    \operatorname{find} ~ x \in \R^n ~ \operatorname{such~that} ~ 0 \in A(x) + B(x) + C(x)\,,
\end{equation}
where $A: \R^n \rightrightarrows \R^n$ and $B: \R^n \rightrightarrows \R^n$ are maximally monotone operators, and $C: \R^n \to \R^n$ is a single-valued $\theta$-cocoercive map. Throughout the paper, we assume that $\operatorname{zer}(A+B+C) \neq \emptyset$. When all three operators are available and the corresponding resolvent or forward computations are reasonably inexpensive, a standard approach to solving~\eqref{eq:monotone_inclusion} relies on the \gls*{DY} splitting~\cite{davis2017three}. This method treats $A$ and $B$ implicitly through separate evaluations of their resolvents, while evaluating $C$ explicitly in a forward step. The \gls*{FPI} associated with the \gls*{DY} splitting is
\begin{subequations}
    \label{eq:davis_yin_splitting}
    \begin{align}
        \label{eq:davis_yin_splitting_x}
        x_t &= J_{\gamma B}(z_t)\\
        \label{eq:davis_yin_splitting_y}
        y_t &= J_{\gamma A}(2x_t - z_t - \gamma C(x_t))\\
        \label{eq:davis_yin_splitting_z}
        z_{t+1} &= z_t + y_t - x_t\,,
    \end{align}
\end{subequations}
where $\gamma > 0$ denotes the step-size, or, equivalently, $z_{t+1} = T_{\mathrm{DY}}(z_t)$, where $T_{\mathrm{DY}} = I - J_{\gamma B} + J_{\gamma A} \left(2 J_{\gamma B} - I - \gamma C J_{\gamma B}\right)$. In particular, for any $\gamma \in (0, 2 \theta)$, the operator $T_{\mathrm{DY}}$ is averaged, the sequence $\{z_t\}_{t \in \N}$ defined by~\eqref{eq:davis_yin_splitting_z} converges to a fixed point of
$T_{\mathrm{DY}}$, while sequences $\{x_t\}_{t \in \N}$ and $\{y_t\}_{t \in \N}$ defined by~\eqref{eq:davis_yin_splitting_x} and~\eqref{eq:davis_yin_splitting_y}, respectively, converge to a solution of~\eqref{eq:monotone_inclusion} for every $z_0 \in \R^n$, see~\cite[Proposition~3.1]{davis2017three} and~\cite[Theorem~1.1]{davis2017three}.

In this paper, we consider the case in which the forward evaluation of the operator $C$ in~\eqref{eq:davis_yin_splitting_y} must be approximated from data. To this end, we assume access to a stochastic oracle $O: \R^n \times \Xi \to \R^n$ and a dataset $\mathcal{D} = \{\xi^1, \dots, \xi^N\}$ of $N \in \N$ \gls*{iid} samples from an \textit{unknown} probability distribution $\mathbb{P}$ over the sample space $\Xi \subseteq \R^d$. Then, we consider the \gls*{SAA} of $C$ in~\eqref{eq:davis_yin_splitting_y} given by
\begin{equation}
\label{eq:data_driven_approximation_C_saa}
	\widehat{C}_{\mathcal{D}}(x) = \frac{1}{N} \sum_{i = 1}^{N} O(x, \xi^{i})\,,
\end{equation}
where $O(x, \xi^i)$ denotes the output obtained by querying the oracle at a point $x \in \R^n$ with the training sample $\xi^i \in \mathcal{D}$.

Next, we postulate few assumptions on the noisy oracle $O$:
\begin{assumption}
    \label{ass:oracle_unbiased}
	For all $x \in \R^n$, $\mathbb{E}_{\xi \sim \mathbb{P}}[{O(x, \xi)}] = C(x)$.
\end{assumption}
\begin{assumption}
\label{ass:oracle_replacement-sensitivity}
    There exists $\varsigma > 0$ such that
    \begin{equation}
    \label{eq:replacement-sensitivity}
        \norm{O(x, \xi) - O(x, \xi')} \leq \varsigma\,,
    \end{equation}
    for every $x \in \R^n$ and every
    $\xi, \xi' \in \Xi$.
\end{assumption}
\begin{assumption}
\label{ass:oracle_theta_cocoercive}
    The map $O(\cdot, \xi)$ is $\theta$-cocoercive for all $\xi \in \Xi$.
\end{assumption}
Assumption~\ref{ass:oracle_unbiased} is common in the statistical learning literature~\cite{tsiamis2023statistical}, ensuring that the oracle $O$ is unbiased. Assumption~\ref{ass:oracle_replacement-sensitivity} instead imposes a uniform upper bound on the sensitivity of the oracle to changes in the stochastic input $\xi$, generalizing~\cite[Assumption~2.2]{fabiani2026finite}. In fact, if $\norm{O(x, \xi)} \leq M$ for every $x \in \R^n$ and $\xi \in \Xi$ as per~\cite[Assumption~2.2]{fabiani2026finite}, then~\eqref{eq:replacement-sensitivity} holds with $\varsigma = 2M$. Last, Assumption~\ref{ass:oracle_theta_cocoercive} formalizes the requirement that the oracle $O$ inherits any structural properties satisfied by the true operator $C$ akin to similar works, e.g.,~\cite{hardt2016train, farnia2021train, fabiani2026finite}.

Leveraging the \gls*{SAA} in \eqref{eq:data_driven_approximation_C_saa}, we define the \gls*{FPI} associated with the data-driven \gls*{DY} splitting as
\begin{subequations}
\label{eq:data_driven_davis_yin_splitting}
\begin{align}
    \label{eq:data_driven_davis_yin_splitting_x}
    \hat{x}_t &= J_{\gamma B}(\hat{z}_t)\\
    \hat{y}_t &= J_{\gamma A}(2\hat{x}_t - \hat{z}_t - \gamma \widehat{C}_{\mathcal{D}}(\hat{x}_t))\\
    \hat{z}_{t+1} &= \hat{z}_t + \hat{y}_t - \hat{x}_t\,,
\end{align}
\end{subequations}
or, equivalently, as $z_{t+1} = \widehat{T}_{\mathcal{D}}(z_t)$, where $\widehat{T}_{\mathcal{D}} = I - J_{\gamma B} + J_{\gamma A} (2 J_{\gamma B} - I - \gamma \widehat{C}_{\mathcal{D}} J_{\gamma B})$ denotes the empirical counterpart of the \gls*{DY} fixed-point operator $T_{\mathrm{DY}}$. In particular, we observe that, unlike~\eqref{eq:davis_yin_splitting}, the \gls*{FPI} in~\eqref{eq:data_driven_davis_yin_splitting} only involves quantities that are either directly available or computable from the samples $\xi^i$, with $i \in \{1, \dots, N\}$. Moreover, conditional on a dataset realization $\mathcal{D} \sim \mathbb{P}^N$, the \gls*{FPI}~\eqref{eq:data_driven_davis_yin_splitting} is deterministic, as no additional randomness is introduced during its execution, and symmetric \gls*{wrt} $\mathcal{D}$, since $\widehat{C}_{\mathcal{D}}$ depends on the samples only through their average and the iterate $\hat{z}_K$ obtained after $K \in \N$ steps of~\eqref{eq:data_driven_davis_yin_splitting} is thus invariant under permutations of $\{\xi^1,\ldots,\xi^N\}$.

Our main research question is then to derive non-asymptotic bounds characterizing, for a given sample size $N$ and iteration budget $K$, how close one can get to a solution of~\eqref{eq:monotone_inclusion} by using the data-driven \gls*{FPI} in~\eqref{eq:data_driven_davis_yin_splitting} instead of the exact \gls*{FPI} in~\eqref{eq:davis_yin_splitting}. More formally, given a confidence parameter $\delta \in (0,1)$, our goal is to determine a radius $\epsilon \geq 0$ ensuring that the data-driven output $\hat{z}_K$ satisfies the \emph{true} \gls*{DY} fixed-point residual bound
\begin{equation}
\label{eq:davis_yin_fixed_point_residual_bound}
    \rho_{\mathrm{DY}}(\hat{z}_K) = \|\hat{z}_{K} - T_{\mathrm{DY}}(\hat{z}_{K})\| \leq \gamma \epsilon\,,
\end{equation}
with probability at least $1 - \delta$ over the random draw $\mathcal{D} \sim \mathbb{P}^N$.

\begin{remark}
    The bound~\eqref{eq:davis_yin_fixed_point_residual_bound} can also be interpreted as a finite-sample certificate for the original inclusion problem~\eqref{eq:monotone_inclusion}. To see this, let $x$ and $y$ denote the points generated by the data-driven output $\hat{z}_K$ through~\eqref{eq:davis_yin_splitting_x} and~\eqref{eq:davis_yin_splitting_y}, respectively, that is,
    \begin{equation*}
        x = J_{\gamma B}(\hat{z}_K)\,, \quad y = J_{\gamma A}(2x - \hat{z}_K - \gamma C(x))\,.
    \end{equation*}
    Then, by~\eqref{eq:davis_yin_splitting_z}, $\norm{x-y} = \norm{\hat{z}_{K} - T_{\mathrm{DY}}(\hat{z}_{K})} \leq \gamma \epsilon$ and by the definition of the resolvents, there exist $a \in A(y)$ and $b \in B(x)$ such that $a + b + C(x) = \frac{x - y}{\gamma}$. Hence, whenever~\eqref{eq:davis_yin_fixed_point_residual_bound} holds, one obtains $\norm{a + b + C(x)} \leq \epsilon$, showing that a single exact \gls*{DY} step from the data-driven output $\hat{z}_K$ yields points $x$ and $y$  that are within $\gamma \epsilon$ of each other and whose associated operator values satisfy the inclusion~\eqref{eq:monotone_inclusion} up to an error of size $\epsilon$. %
\end{remark}

In the next section, we will analyze the data-driven \gls*{FPI}~\eqref{eq:data_driven_davis_yin_splitting} from a statistical learning-theoretic perspective to provide quantitative answers to the research question above.%

\section{Main results}
\label{sec:main_results}

We next establish distribution-free, finite-sample certificates on the quality of the output $\hat{z}_K$ returned by the data-driven \gls*{DY} splitting~\eqref{eq:data_driven_davis_yin_splitting}. To this end, we adopt a statistical learning perspective and interpret the $K$-step data-driven \gls*{FPI}~\eqref{eq:data_driven_davis_yin_splitting} as a learning algorithm $\mathcal{A}_{\mathrm{DY}} : \Xi^N \to \mathcal{H}_{\mathrm{DY}}$ mapping a dataset $\mathcal{D} \in \Xi^N$ to a hypothesis $\omega_K = \mathcal{A}_{\mathrm{DY}}(\mathcal{D})\in\mathcal{H}_{\mathrm{DY}} \subseteq \R^{4n}$, with
\begin{equation*}
\begin{aligned}
    &\mathcal{H}_{\mathrm{DY}}=\\
    &\{\omega = (x, y, z, \zeta) \in \R^{4n}  : x \! = \! J_{\gamma B}(z), ~ y \! = \! J_{\gamma A}(2x \! - \! z \! - \! \gamma \zeta)\}\,.
\end{aligned}
\end{equation*}

Specifically, we define the output of $\mathcal{A}_{\mathrm{DY}}$ by $\omega_K=(\hat{x}_K, \hat{y}_K, \hat{z}_K, \hat{\zeta}_K)$, where $\hat{x}_K$,  $\hat{y}_K$, and $\hat{z}_K$ are the iterates generated by the $K$-step data-driven \gls*{FPI}~\eqref{eq:data_driven_davis_yin_splitting}, and $\hat{\zeta}_K = \widehat{C}_{\mathcal{D}}(\hat{x}_K)$.

Our first result upper bounds the unknown exact \gls*{DY} fixed-point residual $\rho_{\mathrm{DY}}(\hat{z}_K)$ associated with the hypothesis $\omega_K = \mathcal{A}_{\mathrm{DY}}(\mathcal{D})$ in terms of the corresponding empirical fixed-point residual $\hat{\rho}_{\mathcal{D}}(\hat{z}_K) = \|\hat{z}_{K} - \widehat{T}_{\mathcal{D}}(\hat{z}_{K})\|$ and a statistical error term measuring how well the \gls*{SAA} operator $\widehat{C}_{\mathcal{D}}$ in~\eqref{eq:data_driven_approximation_C_saa} approximates the true operator $C$ at the data-dependent point $\hat{x}_K$:

\begin{proposition}
\label{prop:davis_yin_true_residual_decomposition} 
    For any training dataset $\mathcal{D} \in \Xi^N$, the true \gls*{DY} fixed-point residual associated with the hypothesis $\omega_K = \mathcal{A}_{\mathrm{DY}}(\mathcal{D})$ is upper bounded by%
    \begin{equation}
    \label{eq:davis_yin_true_residual_decomposition} 
        \rho_{\mathrm{DY}}(\hat{z}_K) \leq \hat{\rho}_{\mathcal{D}}(\hat{z}_K) + \gamma \|\Delta_{\mathcal{D}}(\hat{x}_K)\|\,,
    \end{equation}
    where $\Delta_{\mathcal{D}}(\hat{x}_K) = \widehat{C}_{\mathcal{D}}(\hat{x}_K) - C(\hat{x}_K) = \hat{\zeta}_K - C(\hat{x}_K)$. 
\end{proposition}
\begin{proof}
    By definition of $\rho_{\mathrm{DY}}(\hat{z}_K)$ in~\eqref{eq:davis_yin_fixed_point_residual_bound} and using the triangle inequality, we have that
    \begin{align}
        \nonumber
        \rho_{\mathrm{DY}}(\hat{z}_K) &= \|\hat{z}_{K} - \widehat{T}_{\mathcal{D}}(\hat{z}_K) + \widehat{T}_{\mathcal{D}}(\hat{z}_K) - T_{\mathrm{DY}}(\hat{z}_{K})\| \\
        \nonumber
        &\leq \|\hat{z}_{K} - \widehat{T}_{\mathcal{D}}(\hat{z}_K)\| + \|\widehat{T}_{\mathcal{D}}(\hat{z}_K) - T_{\mathrm{DY}}(\hat{z}_{K})\| \\
        \label{eq:davis_yin_true_residual_decomposition_intermediate} 
        &= \hat{\rho}_{\mathcal{D}}(\hat{z}_K) + \|\hat{y}_K - \hat{x}_K - (y - x)\|\,,
    \end{align}
    where $x = J_{\gamma B}(\hat{z}_K)$ and $y = J_{\gamma A}(2x - \hat{z}_K - \gamma C(x))$ denote the points generated through~\eqref{eq:davis_yin_splitting_x} and~\eqref{eq:davis_yin_splitting_y}, respectively. Since  $J_{\gamma B}$ is available by assumption, the resolvent evaluations in~\eqref{eq:davis_yin_splitting_x} and~\eqref{eq:data_driven_davis_yin_splitting_x} coincide at $\hat{z}_K$, yielding $\hat{x}_K =x$ in~\eqref{eq:davis_yin_true_residual_decomposition_intermediate}. Finally, since the resolvent $J_{\gamma A}$ of the maximal monotone operator $A$ is (firmly) nonexpansive by \cite[Corollary~23.9]{bauschke2017convex}, we have that
    \begin{equation*}
        \| \hat{y}_K - y \| \leq \| (2\hat{x}_K - \hat{z}_K - \gamma \widehat{C}_{\mathcal{D}}(\hat{x}_K)) - (2x - \hat{z}_K - \gamma C(x))\|\,,
    \end{equation*}
    from which the bound~\eqref{eq:davis_yin_true_residual_decomposition} follows.
\end{proof}

Since the empirical fixed-point residual $\hat{\rho}(\hat{z}_K)$ is computable, Proposition~\ref{prop:davis_yin_true_residual_decomposition} shows that certifying the true residual $\rho_{\mathrm{DY}}(\hat{z}_K)$ reduces to controlling the \gls*{SAA} error $\|\Delta_{\mathcal{D}}(\hat{x}_K)\|$. The main statistical difficulty is that the evaluation point $\hat{x}_K$ is computed from the same dataset $\mathcal{D}$ used to define $\widehat{C}_{\mathcal{D}}$ in~\eqref{eq:data_driven_approximation_C_saa}. Therefore, unlike the case where $\Delta_{\mathcal{D}}$ is evaluated at a deterministic point $x \in \R^n$, the empirical average
\begin{equation}
\label{eq:saa_at_data_dependent_final_iterate_without_norm}
    \Delta_{\mathcal{D}}(\hat{x}_K) = \frac{1}{N} \sum_{i = 1}^N O(\hat{x}_K, \xi^i) - C(\hat{x}_K)\,,
\end{equation}
involves random vectors that are generally neither independent nor individually centered \gls*{wrt} the random draw of $\mathcal{D} \sim \mathbb{P}^N$.

To address this challenge, we bound next the sensitivity of the hypothesis $\omega_K = \mathcal{A}_{\mathrm{DY}}(\mathcal{D})$ to changes in the training dataset $\mathcal{D}$. Specifically, let $\mathcal{D}^i$ denote the dataset obtained by replacing the $i$-th sample in $\mathcal{D}$, $i \in \{1, \dots, N\}$, with an arbitrary $\tilde{\xi}^i \in \Xi$, which is \gls*{iid} \gls*{wrt} samples contained in $\mathcal{D}\setminus\{\xi^i\}$. We then have the following result.%

\begin{proposition}
    \label{prop:one_sample_replacement_sensitivity}
    Let $\gamma \in (0, 2\theta)$. For any pair of neighboring datasets $\mathcal{D} \in \Xi^N$ and $\mathcal{D}^i$, with $i \in \{1, \dots, N\}$, the respective hypotheses $\omega_K = \mathcal{A}_{\mathrm{DY}}(\mathcal{D})$ and $\omega_K^{i} = \mathcal{A}_{\mathrm{DY}}(\mathcal{D}^{i})$ satisfy
    \begin{equation}
    \label{eq:one_sample_replacement_sensitivity}
        \|\hat{x}_K - \hat{x}_K^{i}\| \leq \|\hat{z}_K - \hat{z}_K^{i}\| \leq \frac{\gamma\varsigma K}{N}\,,
    \end{equation}
    whenever their initial points $\hat{z}_0 \in \R^n$ and $\hat{z}_0^i \in \R^n$ coincide. 
\end{proposition}

\begin{proof}
    We begin by showing that, for any training dataset $\mathcal{D} \in \Xi^N$, the operator $\widehat{C}_{\mathcal{D}}$ in~\eqref{eq:data_driven_approximation_C_saa} is $\theta$-cocoercive. For any $x, y \in \R^n$:
    \begin{align*}
        &\langle \widehat{C}_{\mathcal{D}}(x) - \widehat{C}_{\mathcal{D}}(y), x - y \rangle = 
        \frac{1}{N} \sum_{i = 1}^{N} \langle O(x, \xi^i) - O(y, \xi^i), x - y \rangle\\
        &~ \overset{(a)}{\geq} \frac{\theta}{N} \sum_{i = 1}^{N} \| O(x, \xi^i) - O(y, \xi^i) \|^2 \\
        &~ \overset{(b)}{\geq} \theta \norm{\frac{1}{N} \sum_{i = 1}^{N} O(x, \xi^i) - O(y, \xi^i)}^2 = \theta \| \widehat{C}_{\mathcal{D}}(x) - \widehat{C}_{\mathcal{D}}(y) \|^2\,,
    \end{align*}
    where $(a)$ follows from $\theta$-cocoercivity of the oracle $O(\cdot, \xi)$, for all $\xi \in \Xi$ as per Assumption~\ref{ass:oracle_theta_cocoercive}, and $(b)$ from Jensen's inequality. Hence, when $\gamma \in (0, 2\theta)$, we have that the empirical \gls*{DY} fixed-point operator $\widehat{T}_{\mathcal{D}}$ is averaged, and thus nonexpansive, for every $\mathcal{D} \in \Xi^N$ by \cite[Proposition~3.1]{davis2017three}. Thus, for any $t \in \N$ the triangle inequality leaves us with:
    \begin{align}
        \nonumber
        &\|\hat{z}_{t+1} - \hat{z}_{t+1}^{i}\| = \|\widehat{T}_{\mathcal{D}}(\hat{z}_{t}) - \widehat{T}_{\mathcal{D}^i}(\hat{z}_{t}^{i})\| \\
        \nonumber
        &\qquad \leq \|\widehat{T}_{\mathcal{D}}(\hat{z}_{t}) - \widehat{T}_{\mathcal{D}}(\hat{z}_{t}^{i})\| + \| \widehat{T}_{\mathcal{D}}(\hat{z}_{t}^{i}) - \widehat{T}_{\mathcal{D}^i}(\hat{z}_{t}^{i})\|\\
        \label{eq:one_sample_replacement_sensitivity_step_1}
        &\qquad \leq \|\hat{z}_{t} - \hat{z}_{t}^{i}\| + \|\widehat{T}_{\mathcal{D}}(\hat{z}_{t}^{i}) - \widehat{T}_{\mathcal{D}^i}(\hat{z}_{t}^{i})\|\,.
    \end{align}
    Proceeding as in the proof of Proposition~\ref{prop:davis_yin_true_residual_decomposition}, since the operators $\widehat{T}_{\mathcal{D}}$ and $\widehat{T}_{\mathcal{D}^i}$ only differ in the argument of the second resolvent evaluation, where $\widehat{C}_{\mathcal{D}}$ is replaced by $\widehat{C}_{\mathcal{D}^i}$, we then have
    \begin{align}
        \nonumber
        &\lVert\widehat{T}_{\mathcal{D}}(\hat{z}_{t}^{i}) - \widehat{T}_{\mathcal{D}^i}(\hat{z}_{t}^{i})\rVert \leq \gamma \| \widehat{C}_{\mathcal{D}}(\hat{x}_t^i) - \widehat{C}_{\mathcal{D}^{i}}(\hat{x}_t^i)\| \\
        \nonumber
        &\quad = \frac{\gamma}{N} \norm{\sum_{j = 1}^{N} O(\hat{x}_t^i, \xi^j) - \sum_{j = 1, j \neq i}^{N} O(\hat{x}_t^i, \xi^j) - O(\hat{x}_t^i, \tilde{\xi}^i)} \\
        \label{eq:one_sample_replacement_sensitivity_step_2}
        &\quad = \frac{\gamma}{N} \norm{ O(\hat{x}_t^i, \xi^i) -O(\hat{x}_t^i, \tilde{\xi}^i)} \leq \frac{\gamma \varsigma}{N}\,,
    \end{align}
    by the nonexpansiveness of $J_{\gamma A}$, the fact that the training datasets $\mathcal{D}$ and $\mathcal{D}^i$ only differ in the $i$-th sample, and the upper bound on the oracle sensitivity given by Assumption~\ref{ass:oracle_replacement-sensitivity}. By combining~\eqref{eq:one_sample_replacement_sensitivity_step_1} and~\eqref{eq:one_sample_replacement_sensitivity_step_2}, the rightmost inequality in~\eqref{eq:one_sample_replacement_sensitivity} is then obtained by induction, since $\omega_K$ and $\omega_K^i$ are generated starting from the same initial point $\hat{z}_0 = \hat{z}_0^i \in \R^n$. Finally, the proof is concluded by observing that
    \begin{equation*}
        \|\hat{x}_K - \hat{x}_K^{i}\| = \|J_{\gamma B}(\hat{z}_K) - J_{\gamma B}(\hat{z}_K^{i})\| \leq \|\hat{z}_K - \hat{z}_K^{i}\|\,,
    \end{equation*}
    since $J_{\gamma B}$ is nonexpansive by \cite[Corollary~23.9]{bauschke2017convex}.
\end{proof}
Let $\eta_K = \max_{i \in \{1, \dots, N\}} ~ \sup_{\mathcal{D}, \mathcal{D}^i} ~ \|\hat{x}_K - \hat{x}_K^{i}\|$ be the maximum distance between the $x$-components of the hypotheses $\omega_K$ and $\omega_K^i$ corresponding to any two neighboring datasets $\mathcal{D}$ and $\mathcal{D}^i$. 
Proposition~\ref{prop:one_sample_replacement_sensitivity} shows that $\eta_K \leq \frac{\gamma\varsigma K}{N}$, that is, $\eta_K$ is bounded, grows at most linearly with the iteration budget $K$, and decays at least as fast as $\frac{1}{N}$ with the number of samples in $\mathcal{D}$.
In what follows, we leverage the replace-one sensitivity bound $\eta_K$ to derive two different distribution-free, finite-sample certificates for the \gls*{SAA} error term $\|\Delta_{\mathcal{D}}(\hat{x}_K)\|$ in~\eqref{eq:davis_yin_true_residual_decomposition} using algorithmic stability arguments~\cite{bousquet2002stability, kutin2002almost, bousquet2020sharper} and martingale inequalities~\cite{boucheron2013concentration, doob1953stochastic, pinelis1994optimum}.

\subsection{A posteriori certificates via algorithmic stability}

In this subsection, we show that the \gls*{SAA} error $\|\Delta_{\mathcal{D}}(\hat{x}_K)\|$ is upper bounded by the population average, evaluated at $\omega_K$, of the scalar loss function $\ell_{\Delta} : \mathcal{H}_{\mathrm{DY}} \times \Xi \to \R$ defined by 
\begin{equation}
\label{eq:davis_yin_scalar_loss_function}
    \ell_{\Delta}(\omega, \xi) = \|\zeta - O(x, \xi)\|\,.    
\end{equation}
Leveraging the replace-one sensitivity bound of Proposition~\ref{prop:one_sample_replacement_sensitivity} and inspired by~\cite{fabiani2026finite}, we then prove uniform stability~\cite{bousquet2002stability} of $\mathcal{A}_{\mathrm{DY}}$ \gls*{wrt} $\ell_{\Delta}$, which in turn yields a posteriori finite-sample guarantees for the data-driven \gls*{FPI}~\eqref{eq:data_driven_davis_yin_splitting}.

The following lemma formalizes the first step.
\begin{lemma}
\label{le:saa_error_upper_bounded_by_population_risk}
    For any training dataset $\mathcal{D} \in \Xi^N$, the \gls*{SAA} error at the $x$-component of the  hypothesis $\omega_K = \mathcal{A}_{\mathrm{DY}}(\mathcal{D})$ satisfies
    \begin{equation*}
    \label{eq:saa_error_upper_bounded_by_population_risk}
        \| \Delta_{\mathcal{D}}(\hat{x}_K) \| \leq \mathbb{E}_{\xi \sim \mathbb{P}}[\ell_{\Delta}(\omega_K, \xi)]\,.
    \end{equation*}
\end{lemma}
\begin{proof}
    By definition of $\Delta_{\mathcal{D}}(\hat{x}_K)$ in Proposition~\ref{prop:davis_yin_true_residual_decomposition} and since the stochastic oracle $O$ is unbiased by Assumption~\ref{ass:oracle_unbiased}, we have
    \begin{equation*}
        \Delta_{\mathcal{D}}(\hat{x}_K) = \hat{\zeta}_K - C(\hat{x}_K) = \hat{\zeta}_K - \mathbb{E}_{\xi \sim \mathbb{P}}[O(\hat{x}_K, \xi)]\,.
    \end{equation*}
    Moreover, since the hypothesis $\omega_K = \mathcal{A}_{\mathrm{DY}}(\mathcal{D})$ is fixed when taking the expectation over an independent test sample $\xi \sim \mathbb{P}$ and applying Jensen's inequality, we obtain
    \begin{align}
        \label{eq:saa_error_upper_bounded_by_population_risk_jensen_conservative_step}
        \| \Delta_{\mathcal{D}}(\hat{x}_K) \| &= \|\mathbb{E}_{\xi \sim \mathbb{P}}[\hat{\zeta}_K - O(\hat{x}_K, \xi)] \| \\
        \nonumber
        &\leq \mathbb{E}_{\xi \sim \mathbb{P}}[\| \hat{\zeta}_K - O(\hat{x}_K, \xi) \|] = \mathbb{E}_{\xi \sim \mathbb{P}}[\ell_{\Delta}(\omega_K, \xi)]\,,
    \end{align}
    and the claim follows.
\end{proof}

Together with
Proposition~\ref{prop:davis_yin_true_residual_decomposition}, Lemma~\ref{le:saa_error_upper_bounded_by_population_risk} yields the bound
\begin{equation}
\label{eq:true_residual_upper_bounded_by_empirical_residual_plus_population_risk}
    \rho_{\mathrm{DY}}(\hat z_K) \leq \hat\rho(\hat z_K) + \gamma \mathbb E_{\xi\sim\mathbb P}
    [ \ell_{\Delta}(\omega_K,\xi) ]\,,
\end{equation}
for every hypothesis $\omega_K=\mathcal A_{\mathrm{DY}}(\mathcal D)$. As the probability distribution $\mathbb{P}$ of $\xi$ is unknown, however, the bound above proves not practical as the population average in the right-hand side of~\eqref{eq:true_residual_upper_bounded_by_empirical_residual_plus_population_risk} cannot be evaluated directly. With the goal of using standard generalization bounds to control the gap between the population average in~\eqref{eq:true_residual_upper_bounded_by_empirical_residual_plus_population_risk}
and its empirical counterpart, our next result propagates the replace-one sensitivity estimate $\eta_K$ from the iterate level to the loss level.

\begin{lemma}
\label{le:davis_yin_uniform_stability}
    Let $\gamma \in (0, 2\theta)$. For any pair of neighboring datasets $\mathcal{D} \in \Xi^N$ and $\mathcal{D}^i$, $i \in \{1, \dots, N\}$ and any $\xi \in \Xi$, the respective hypotheses $\omega_K = \mathcal{A}_{\mathrm{DY}}(\mathcal{D})$ and $\omega_K^{i} = \mathcal{A}_{\mathrm{DY}}(\mathcal{D}^{i})$ satisfy
    \begin{equation}
    \label{eq:davis_yin_uniform_stability}
        | \ell_{\Delta}(\omega_K, \xi) - \ell_{\Delta}(\omega_K^i, \xi) | \leq \frac{2\eta_K}{\theta} + \frac{\varsigma}{N} =: \beta_K\,, %
    \end{equation}
    whenever their initial points $z_0 \in \R^n$ and $z_0^i \in \R^n$ coincide. 
\end{lemma}
\begin{proof}
    By definition of $\ell_{\Delta}$ in~\eqref{eq:davis_yin_scalar_loss_function} and by the reverse triangle inequality followed by the triangle inequality, we first upper bound the left-hand side of~\eqref{eq:davis_yin_uniform_stability} by
    \begin{align}
        \nonumber
        &\| \hat{\zeta}_K - O(\hat{x}_K, \xi) - ( \hat{\zeta}_K^i - O(\hat{x}_K^i, \xi) )\|\\ 
        \label{eq:davis_yin_uniform_stability_step_1}
        &\qquad \leq \| \hat{\zeta}_K - \hat{\zeta}_K^i\| + \| O(\hat{x}_K, \xi) - O(\hat{x}_K^i, \xi)\| \,.
    \end{align}
    Recall that the stochastic oracle $O$ is $\theta$-cocoercive by Assumption~\ref{ass:oracle_theta_cocoercive}. Since $\theta$-cocoercivity implies $\frac{1}{\theta}$-Lipschitzness, see, e.g, \cite[Remark~4.15]{bauschke2017convex} and by definition of $\eta_K$, we then have
    \begin{equation}
        \label{eq:davis_yin_uniform_stability_step_2}
        \| O(\hat{x}_K, \xi) - O(\hat{x}_K^i, \xi)\| \leq \frac{1}{\theta} \| \hat{x}_K - \hat{x}_K^i\| \leq \frac{\eta_K}{\theta}\,,
    \end{equation}
    Similarly, by~\eqref{eq:data_driven_approximation_C_saa} and since the average of $\frac{1}{\theta}$-Lipschitz functions is itselft $\frac{1}{\theta}$-Lipschitz, we obtain 
    \begin{align}
        \nonumber
        &\| \hat{\zeta}_K - \hat{\zeta}_K^i\| = \| \widehat{C}_\mathcal{D}(\hat{x}_K) - \widehat{C}_{\mathcal{D}^i}(\hat{x}_K^i) \| \\
        \nonumber
        &\qquad \leq \| \widehat{C}_\mathcal{D}(\hat{x}_K) - \widehat{C}_\mathcal{D}(\hat{x}_K^i)\| + \| \widehat{C}_\mathcal{D}(\hat{x}_K^i) - \widehat{C}_{\mathcal{D}^i}(\hat{x}_K^i) \| \\
        \label{eq:davis_yin_uniform_stability_step_3}
        &\qquad \leq \frac{\eta_K}{\theta} + \frac{\varsigma}{N}\,,
    \end{align}
    where the last inequality follows by arguments analogous to those used in~\eqref{eq:one_sample_replacement_sensitivity_step_2}, since the training datasets $\mathcal{D}$ and $\mathcal{D}^i$ only differ in the $i$-th sample. The proof is concluded by combining the bounds in~\eqref{eq:davis_yin_uniform_stability_step_1},~\eqref{eq:davis_yin_uniform_stability_step_2}, and~\eqref{eq:davis_yin_uniform_stability_step_3}.%
\end{proof}

Lemma~\ref{le:davis_yin_uniform_stability} establishes uniform stability of $\mathcal{A}_{\mathrm{DY}}$ \gls*{wrt} the scalar loss function $\ell_{\Delta}$  in~\eqref{eq:davis_yin_scalar_loss_function}, in the sense of~\cite{bousquet2002stability, kutin2002almost, bousquet2020sharper}. Intuitively,~\eqref{eq:davis_yin_uniform_stability} shows that changing one sample $\xi \in \mathcal{D}$ does not significantly alter the output of $\mathcal{A}_{\mathrm{DY}}$ \gls*{wrt} the metric identified by $\ell_{\Delta}$. 

We now convert the underlying deterministic stability certificate into a high-probability generalization bound.

\begin{theorem}
\label{th:finite_sample_guarantees_algorithmic_stability}
    Let $\gamma \in (0, 2\theta)$ and fix any $\delta \in (0,1)$. Then, with probability at least $1 - \delta$ over the random draw $\mathcal{D} \sim \mathbb{P}^N$, the true \gls*{DY} fixed-point residual associated with the hypothesis $\omega_K = \mathcal{A}_{\mathrm{DY}}(\mathcal{D})$ is upper bounded by %
    \begin{equation}
    \label{eq:finite_sample_guarantees_algorithmic_stability}
        \rho_{\mathrm{DY}}(\hat{z}_K) \leq \hat{\rho}_{\mathcal{D}}(\hat{z}_K) + \gamma (\hat{\ell}_{\mathcal{D}}(\omega_K) + \kappa(N, K, \delta))\,,
    \end{equation}
    where $\hat{\ell}_{\mathcal{D}}(\omega_K) = \frac{1}{N}\sum_{i = 1}^N \ell_{\Delta}(\omega_K, \xi^i)$ is the empirical counterpart of the population average in~\eqref{eq:true_residual_upper_bounded_by_empirical_residual_plus_population_risk} and
    \begin{equation}
    \label{eq:generalization_correction_term_algorithmic_stability}
        \kappa(N, K, \delta) = 2 \beta_K + (4N\beta_K + \varsigma)\sqrt{\frac{\operatorname{log}(1/\delta)}{2N}}\,,
    \end{equation}
    with $\beta_K = \frac{2\eta_K}{\theta} + \frac{\varsigma}{N}$, is a generalization correction term.
\end{theorem}
\begin{proof}
    The claim follows by applying the results of~\cite[Theorem~3.2]{kutin2002almost} to bound the population average $\mathbb E_{\xi\sim\mathbb P}
    [ \ell_{\Delta}(\omega_K,\xi) ]$ in~\eqref{eq:true_residual_upper_bounded_by_empirical_residual_plus_population_risk}. Since $\mathcal{A}_{\mathrm{DY}}$ is $\beta_K$-uniformly stable \gls*{wrt} $\ell_{\Delta}$ by Lemma~\ref{le:davis_yin_uniform_stability}, we are left with verifying that, for any $\mathcal{D} \in \Xi^N$ and $\xi \in \Xi$, the loss $\ell_{\Delta}(\omega_K, \xi)$ is bounded from above. By the triangle inequality and using Assumption~\ref{ass:oracle_replacement-sensitivity}, we have that
    \begin{align*}
        0 \leq \ell_{\Delta}(\omega_K, \xi) &= \| \hat{\zeta}_K - O(\hat{x}_K, \xi)\| \\
        &= \norm{ \frac{1}{N} \sum_{i=1}^N O(\hat{x}_K, \xi^i) - O(\hat{x}_K, \xi)} \\
        &\leq \frac{1}{N} \sum_{i=1}^N \norm{O(\hat{x}_K, \xi^i) - O(\hat{x}_K, \xi)} \leq \varsigma\,.
    \end{align*}
    Hence,~\cite[Theorem~3.2]{kutin2002almost} applies, concluding the proof.
\end{proof}

Theorem~\ref{th:finite_sample_guarantees_algorithmic_stability} establishes our first finite-sample certificate for the data-driven \gls*{DY} splitting method in~\eqref{eq:data_driven_davis_yin_splitting}. Specifically, the generalization bound~\eqref{eq:finite_sample_guarantees_algorithmic_stability} shows that~\eqref{eq:davis_yin_fixed_point_residual_bound} holds with arbitrarily high confidence $1-\delta$ and regardless of the probability distribution $\mathbb{P}$ underlying the data if the radius $\epsilon$ is such that
\begin{equation}
\label{eq:epsilon_value_for_davis_yin_fixed_point_residual_bound_algorithimic_stability}
    \epsilon \geq \frac{\hat{\rho}_{\mathcal{D}}(\hat{z}_K)}{\gamma} 
    + \hat{\ell}_{\mathcal{D}}(\omega_K) + \kappa(N, K, \delta) =: \epsilon_{\mathrm{AS}}(N, K, \delta)\,.   
\end{equation}
Up to a scaling factor $\gamma$, the first term in~\eqref{eq:epsilon_value_for_davis_yin_fixed_point_residual_bound_algorithimic_stability} is the empirical counterpart of the \gls*{DY} fixed-point residual $\rho_{\mathrm{DY}}(\hat{z}_K)$ after $K$ iterations. The second and third addends are instead statistical error terms. In particular, by inspection of~\eqref{eq:generalization_correction_term_algorithmic_stability} and since $\eta_K \leq \frac{\gamma\varsigma K}{N}$ by Proposition~\ref{prop:one_sample_replacement_sensitivity}, we observe that, for any fixed $K \in \N$, the stability parameter $\beta_K$ decays at rate $\frac{1}{N}$ with the number of samples in $\mathcal{D}$, and the generalization correction term $\kappa(N, K, \delta)$ thus vanishes at rate $\frac{1}{\sqrt{N}}$, up to logarithmic
factors in $\frac{1}{\delta}$.\footnote{More generally, if the iteration budget is allowed to depend on the sample size, the correction term $\kappa(N, K, \delta)$ still vanishes provided that $K = o(\sqrt{N})$.} Conversely, as shown by Example~\ref{ex:average_oracle_discrepancy_does_not_vanish} below, the term $\hat{\ell}_{\mathcal{D}}(\omega_K)$, which measures the average sample-wise oracle dispersion around the \gls*{SAA}~\eqref{eq:data_driven_approximation_C_saa} at $\hat{x}_K$, need not converge to zero as $N$ increases. Therefore, contrary to~\cite{fabiani2026finite}, we conclude that the algorithmic stability analysis and the bound~\eqref{eq:epsilon_value_for_davis_yin_fixed_point_residual_bound_algorithimic_stability} generally do not yield an asymptotically consistent residual certificate.

\begin{example}
\label{ex:average_oracle_discrepancy_does_not_vanish}
    Let us consider an instance of~\eqref{eq:monotone_inclusion} on $\R$, where $A = 0$, $B = 0$, and $C = \operatorname{tanh}$. We note that $A$ and $B$ are maximally monotone, while $C$ is bounded and $1$-cocoercive, so~\eqref{eq:davis_yin_splitting} reduces to the \gls*{FB} in \cite[Algorithm~1]{fabiani2026finite}. Moreover, the inclusion~\eqref{eq:monotone_inclusion} has the unique solution $x^\star = 0$. 
    
    For any $\tau > 0$, consider the oracle $O(x, \xi) = \operatorname{tanh}(x) + \tau \xi$, where $\xi$ is a Rademacher random variable taking values $+1$ and $-1$ with equal probability. Then, $\mathbb{E}_{\xi \sim \mathbb{P}}[O(x, \xi)] = \operatorname{tanh}(x) = C(x)$, so the oracle $O$ is unbiased. Moreover,~\cite[Assumption~2.2]{fabiani2026finite}, Assumption~\ref{ass:oracle_replacement-sensitivity}, and Assumption~\ref{ass:oracle_theta_cocoercive} hold with $M = \tau$, $\varsigma = 2\tau$, and $\theta = 1$, respectively. 

    Let $N \in \N$ be even and consider a balanced dataset $\mathcal{D} = \{\xi^1, \dots, \xi^N\}$ where the values $+1$ and $-1$ each appear $\frac{N}{2}$ times. Then, the \gls*{SAA}~\eqref{eq:data_driven_approximation_C_saa} gives $\widehat C_{\mathcal D}(x) = \frac{1}{N} \sum_{i=1}^N O(x,\xi^i) = \tanh(x) = C(x)$ for every $x \in \R$, and~\eqref{eq:davis_yin_splitting} and~\eqref{eq:data_driven_davis_yin_splitting} therefore coincide. 
    In particular, if $\hat{z}_0 = 0$, then~\eqref{eq:data_driven_davis_yin_splitting} yields $\hat x_t=\hat y_t=\hat z_t=0$ at all times, and thus $\rho_{\mathrm{DY}}(\hat{z}_K) = \hat{\rho}_{\mathcal{D}}(\hat{z}_K) = 0$ and $\Delta_{\mathcal{D}}(\hat{x}_K) = 0$. Nevertheless, for every $K \in \N$, the right-hand side of~\eqref{eq:epsilon_value_for_davis_yin_fixed_point_residual_bound_algorithimic_stability} can be arbitrarily large since
    \begin{equation*}
        \hat{\ell}_{\mathcal{D}}(\omega_K) = \frac{1}{N}\sum_{i = 1}^N \ell_{\Delta}(\omega_K, \xi^i) = \frac{1}{N}\sum_{i = 1}^N \|- \tau \xi^i\| = \tau\,,
    \end{equation*}
    and $\tau > 0$ is arbitrary, meaning that the radius $\epsilon$ in~\eqref{eq:epsilon_value_for_davis_yin_fixed_point_residual_bound_algorithimic_stability} cannot generally be made arbitrarily small by increasing $N$ and $K$. This is because, while the \gls*{SAA}~\eqref{eq:data_driven_approximation_C_saa} and~\eqref{eq:saa_at_data_dependent_final_iterate_without_norm} involve an average of signed oracle errors, the surrogate loss $\ell_{\Delta}$ is defined by taking norms before averaging in~\eqref{eq:saa_error_upper_bounded_by_population_risk_jensen_conservative_step}, preventing cancellations among oracle errors with opposite signs. The same phenomenon appears in the \gls*{FB} certificate of \cite{fabiani2026finite}. In fact, the expression for $\epsilon$ given in the proof of Theorem~3.4 in \cite{fabiani2026finite} contains the empirical sample-wise loss $\frac{1}{\gamma N} \sum_{i=1}^{N} \| x_{K+1} - \gamma O(x_{K+1},\xi^{i}) -y_K \|$, in addition to nonnegative stability correction terms. Hence, in the present example, this gives $\epsilon \geq %
        \frac{1}{\gamma N} \sum_{i=1}^N \| - \gamma\tau\xi^i\| = \tau$.
\end{example}

\subsection{A priori certificates via martingale inequalities}
Motivated by the limitation of the finite-sample certificate~\eqref{eq:epsilon_value_for_davis_yin_fixed_point_residual_bound_algorithimic_stability} highlighted by Example~\ref{ex:average_oracle_discrepancy_does_not_vanish}, in this subsection we derive a complementary, a priori finite-sample certificate on the \gls*{SAA} error $\|\Delta_{\mathcal{D}}(\hat{x}_K)\|$ by controlling the random variable~\eqref{eq:saa_at_data_dependent_final_iterate_without_norm} directly, rather than through the scalar surrogate $\ell_{\Delta}$.

As discussed in Section~\ref{sec:main_results}, the main statistical difficulty is that~\eqref{eq:saa_at_data_dependent_final_iterate_without_norm} is not an average of \gls*{iid} centered random vectors \gls*{wrt} the multi-sample extraction of $\mathcal{D} \sim \mathbb{P}^N$. To address this challenge, we consider the Doob martingale $\{M_0, \dots, M_N\}$ associated with the \gls*{SAA} error~\eqref{eq:saa_at_data_dependent_final_iterate_without_norm}, that is, the sequence of best predictions of $\Delta_{\mathcal{D}}(\hat{x}_K)$ after observing the first $i \in \{0, \dots, N\}$ samples only. Then, we leverage the replace-one sensitivity bound $\eta_K$ to control the initial value $\|M_0\|$ and the martingale increments $\| M_i - M_{i-1} \|$, for all $i \in \{1, \dots, N\}$, and conclude by applying the available martingale concentration inequalities established in~\cite{pinelis1994optimum}.

We begin by formalizing the definition of the Doob martingale $\{M_0, \dots, M_N\}$. To this end, we observe that, letting $\tilde{\xi} \sim \mathbb{P}$ denote a new sample drawn independently from $\mathcal{D}$ and by triangle inequality, Assumption~\ref{ass:oracle_unbiased}, and Assumption~\ref{ass:oracle_replacement-sensitivity}:
\begin{align*}
    \| \Delta_{\mathcal{D}}(\hat{x}_K) \| &= \norm{\frac{1}{N} \sum_{i = 1}^N O(\hat{x}_K, \xi^i) - C(\hat{x}_K)}\\
    &= \norm{\frac{1}{N} \sum_{i = 1}^N O(\hat{x}_K, \xi^i) - \mathbb{E}_{\tilde{\xi} \sim \mathbb{P}}[O(\hat{x}_K, \tilde{\xi})]}\\
    &\leq \frac{1}{N} \sum_{i = 1}^N \mathbb{E}_{\tilde{\xi} \sim \mathbb{P}}[ \| O(\hat{x}_K, \xi^i) - O(\hat{x}_K, \tilde{\xi}) \| ] \leq \varsigma\,.
\end{align*}
Therefore, $\Delta_{\mathcal{D}}(\hat{x}_K)$ is integrable, and we can define the Doob martingale $\{M_0, \dots, M_N\}$ by conditioning on the information revealed by the first $i \in \{0, \dots, N\}$ samples~\cite{boucheron2013concentration, doob1953stochastic}, i.e.,
\begin{equation}
\label{eq:doob_martingale_definition_Mi_best_predictions_given_samples_up_to_i}
    M_i = \mathbb{E}_{\xi^{i+1} \sim \mathbb{P}, \dots, \xi^N \sim \mathbb{P}}[\Delta_\mathcal{D}(\hat{x}_K) \mid \xi^1, \dots, \xi^i]\,, 
\end{equation}
with the convention that $M_0 = \mathbb{E}_{\mathcal{D} \sim \mathbb{P}^N} [\Delta_{\mathcal{D}}(\hat{x}_K)]$ and $M_N = \Delta_\mathcal{D}(\hat{x}_K)$. In particular, the sequence $\{M_0, \dots, M_N\}$ satisfies the martingale property by construction, since the law of total expectation yields
\begin{align*}
    &\mathbb{E}_{\xi^i \sim \mathbb{P}} [M_i \mid \xi^1,\ldots,\xi^{i-1}] \\
    &= \mathbb E_{\xi^i\sim\mathbb P} \!
    \left[
        \mathbb{E}_{\xi^{i+1} \sim \mathbb{P}, \ldots, \xi^N \sim \mathbb{P}} \!
        \left[
            \Delta_{\mathcal{D}}(\hat{x}_K)
            \!\mid\!
            \xi^1,{\ldots},\xi^i
        \right]
        \middle|
        \xi^1,{\ldots},\xi^{i-1}
    \right]                                          \\
    &= \mathbb{E}_{\xi^{i} \sim \mathbb{P}, \ldots, \xi^N \sim \mathbb{P}}
    \left[
        \Delta_{\mathcal{D}}(\hat{x}_K)
        \mid
        \xi^1,\ldots,\xi^{i-1}
    \right] = M_{i-1}\,.
\end{align*}

The following lemma provides a deterministic upper bound on the mean of the \gls*{SAA} error in~\eqref{eq:saa_at_data_dependent_final_iterate_without_norm}.
    
\begin{lemma}
    \label{le:saa_error_bias_control}
    Let $\gamma \in (0, 2\theta)$ and $\mathcal{D} \sim \mathbb{P}^N$. Then, the initial value of the Doob martingale $\{M_0, \dots, M_N\}$ in~\eqref{eq:doob_martingale_definition_Mi_best_predictions_given_samples_up_to_i} satisfies
    \begin{equation}
    \label{eq:saa_error_bias_control}
        \norm{\mathbb{E}_{\mathcal{D} \sim \mathbb{P}^N} [\Delta_{\mathcal{D}}(\hat{x}_K)]} = \| M_0 \| \leq \frac{2\eta_K}{\theta}\,.
    \end{equation}
\end{lemma}
\begin{proof}
    With slight abuse of notation, for any $i \in \{1, \dots, N\}$, let $\mathcal{D}^i$ be any dataset obtained by replacing the $i$-th sample in $\mathcal{D}$ with a new sample $\tilde{\xi}^i \sim \mathbb{P}$, drawn independently from $\mathcal{D}$ (rather than any new sample $\tilde{\xi}^i \in \Xi$). Moreover, let $\omega_K = \mathcal{A}_{\mathrm{DY}}(\mathcal{D})$, and $\omega_K^i = \mathcal{A}_{\mathrm{DY}}(\mathcal{D}^i)$.
    By Assumption~\ref{ass:oracle_unbiased}, since $\hat{x}_K^i$ is independent of $\xi^i$, we have that
    \begin{equation}
    \label{eq:saa_error_bias_control_step_1}
        \mathbb{E}_{\xi^i \sim \mathbb{P}}[O(\hat x_K^i,\xi^i) - C(\hat x_K^i) | \mathcal{D}^i] = 0\,.
    \end{equation}
    Hence, recalling the definition of $\Delta_{\mathcal{D}}(\hat{x}_K)$ in~\eqref{eq:saa_at_data_dependent_final_iterate_without_norm} and using the linearity of the expectation operator, we then have that
    \begin{align*}
        &\mathbb{E}_{\mathcal{D} \sim \mathbb{P}^N} [\Delta_{\mathcal{D}}(\hat{x}_K)] = \frac{1}{N} \sum_{i = 1}^N \mathbb{E}_{\mathcal{D} \sim \mathbb{P}^N}[ O(\hat{x}_K, \xi^i) - C(\hat{x}_K)]\\
        &= \frac{1}{N} \sum_{i=1}^N \mathbb E_{\substack{\mathcal{D} \sim \mathbb{P}^N \\ \tilde{\xi}^i \sim \mathbb{P}}} [ O(\hat x_K,\xi^i) \!-\! C(\hat x_K) \!-\! O(\hat x_K^i,\xi^i) \!+\! C(\hat x_K^i)]\,,
    \end{align*}
    where the last identity follows by~\eqref{eq:saa_error_bias_control_step_1} and the law of total expectation. Taking norms and using Jensen's inequality and the triangle inequality, we the conclude that the left-hand side of~\eqref{eq:saa_error_bias_control} is upper bounded by
    \begin{align*}
        &\frac{1}{N} \! \sum_{i=1}^N \mathbb E_{\substack{\mathcal{D} \sim \mathbb{P}^N \\ \tilde{\xi}^i \sim \mathbb{P}}} [ \|O(\hat x_K,\xi^i) \!-\! O(\hat x_K^i,\xi^i) \| \!+\! \|C(\hat x_K) \!-\! C(\hat x_K^i) \|] \\
        &\quad \leq
        \frac{2}{\theta N}
        \sum_{i=1}^N
        \mathbb E_{\substack{\mathcal{D} \sim \mathbb{P}^N \\ \tilde{\xi}^i \sim \mathbb{P}}}
        \left[
            \left\|
                \hat x_K-\hat x_K^i
            \right\|
        \right] \leq \frac{2\eta_K}{\theta}\,,
    \end{align*}
    where, similarly to~\eqref{eq:davis_yin_uniform_stability_step_2} and~\eqref{eq:davis_yin_uniform_stability_step_3}, the last inequalities follow since both the oracle $O$ and the operator $C$ are $\theta$-cocoercive and thus $\frac{1}{\theta}$-Lipschitz by~\cite[Remark~4.15]{bauschke2017convex}. 
\end{proof}

We continue our analysis by showing that the replace-one sensitivity estimate $\eta_K$ also yields a uniform bound on the
martingale increments as follows:

\begin{lemma}
\label{le:doob_martingale_increment_bounded}
    Let $\gamma \in (0, 2\theta)$ and $\mathcal{D} \sim \mathbb{P}^N$. Then, the increments of the Doob martingale $\{M_0, \dots, M_N\}$ in~\eqref{eq:doob_martingale_definition_Mi_best_predictions_given_samples_up_to_i} satisfy
    \begin{equation}
    \label{eq:doob_martingale_increment_bounded}
        \| M_i - M_{i-1}\| \leq \frac{2\eta_K}{\theta} + \frac{\varsigma}{N} =: \nu_K \,, %
    \end{equation}    
    almost surely for all $i \in \{1, \dots, N\}$.
\end{lemma}
\begin{proof}
    Proceeding as in the proof of Lemma~\ref{le:saa_error_bias_control}, let $\mathcal{D}^i$ be any dataset obtained by replacing the $i$-th sample in $\mathcal{D}$ with a new sample $\tilde{\xi}^i \sim \mathbb{P}$, drawn independently from $\mathcal{D}$, and define $\omega_K = \mathcal{A}_{\mathrm{DY}}(\mathcal{D})$ and $\omega_K^i = \mathcal{A}_{\mathrm{DY}}(\mathcal{D}^i)$. By the definition of the Doob martingale $\{M_0, \dots, M_N\}$ in~\eqref{eq:doob_martingale_definition_Mi_best_predictions_given_samples_up_to_i}, we have that
    \begin{align*}
        M_{i-1} &=\mathbb{E}_{\xi^{i} \sim \mathbb{P}, \dots, \xi^N \sim \mathbb{P}}
        \left[
            \Delta_{\mathcal{D}}(\hat{x}_K) \mid \xi^1,\dots,\xi^{i-1}
        \right]\\
        &=\mathbb{E}_{\tilde{\xi}^{i} \sim \mathbb{P}, \dots, \xi^N \sim \mathbb{P}}
        \left[
            \Delta_{\mathcal{D}^i}(\hat{x}_K^i) \mid \xi^1,\dots,\xi^{i-1}
        \right]\\
        &=\mathbb{E}_{\tilde{\xi}^{i} \sim \mathbb{P}, \dots, \xi^N \sim \mathbb{P}}
        \left[
            \Delta_{\mathcal{D}^i}(\hat{x}_K^i) \mid \xi^1,\dots,\xi^{i}
        \right]\,,
    \end{align*}
    where the last identity follows because $\Delta_{\mathcal{D}^i}(\hat{x}_K^i)$ does not depends on $\xi^i$. Therefore, the martingale increments satisfy
    \begin{align*}
        &M_i-M_{i-1} = \mathbb{E}_{\xi^{i+1} \sim \mathbb{P}, \dots, \xi^N \sim \mathbb{P}}
        \left[
            \Delta_{\mathcal{D}}(\hat{x}_K)
            \mid \xi^1, \dots, \xi^i
        \right] \\
        &\qquad\qquad\qquad\quad  - \mathbb{E}_{\tilde{\xi}^{i} \sim \mathbb{P}, \dots, \xi^N \sim \mathbb{P}}
        \left[
            \Delta_{\mathcal{D}^i}(\hat{x}_K^i) \mid \xi^1,\dots,\xi^{i}
        \right] \\
        & =
        \mathbb{E}_{\tilde{\xi}^i \sim \mathbb{P}, \xi^{i+1} \sim \mathbb{P}, \dots, \xi^N \sim \mathbb{P}}
        \left[
            \Delta_{\mathcal D}(\hat{x}_K) -\Delta_{\mathcal{D}^i}(\hat{x}_K^i)
            \mid
            \xi^1,\ldots,\xi^i
        \right]\,.
    \end{align*}
    By applying the Jensen's inequality, we then obtain that the left-hand side of~\eqref{eq:doob_martingale_increment_bounded} is upper bounded by
    \begin{equation*}
        \mathbb{E}_{\tilde{\xi}^i \sim \mathbb{P}, \xi^{i+1} \sim \mathbb{P}, \dots, \xi^N \sim \mathbb{P}}
        \left[ \|
            \Delta_{\mathcal D}(\hat{x}_K) -\Delta_{\mathcal{D}^i}(\hat{x}_K^i) \|
            \mid
            \xi^1,\ldots,\xi^i
        \right]\,.
    \end{equation*}
    We continue by bounding the term $\| \Delta_{\mathcal D}(\hat{x}_K) -\Delta_{\mathcal{D}^i}(\hat{x}_K^i) \|$ inside the conditional expectation. By definition of the \gls*{SAA} error in~\eqref{eq:saa_at_data_dependent_final_iterate_without_norm} and using the triangle inequality, we obtain
    \begin{align*}
        &\| \Delta_{\mathcal D}(\hat{x}_K) -\Delta_{\mathcal{D}^i}(\hat{x}_K^i) \| \leq \| \widehat C_{\mathcal D}(\hat x_K) - \widehat C_{\mathcal{D}^{i}}(\hat x_K^i) \|\\
        &\qquad\qquad\qquad\qquad\qquad\qquad\qquad + \| C(\hat x_K) - C(\hat x_K^i) \|\,,
    \end{align*}
    which yields $\|M_i-M_{i-1}\| \leq \frac{\eta_K}{\theta} + \frac{\varsigma}{N} + \frac{\eta_K}{\theta}$ almost surely by~\eqref{eq:davis_yin_uniform_stability_step_3} and $\frac{1}{\theta}$-Lipschitzness of $C$ as per~\cite[Remark~4.15]{bauschke2017convex}.
\end{proof}

We are ready to present our a priori finite-sample certificate on the \gls*{SAA} error term $\|\Delta_{\mathcal{D}}(\hat{x}_K)\|$ by combining the results of Lemma~\ref{le:saa_error_bias_control} and Lemma~\ref{le:doob_martingale_increment_bounded} with the martingale concentration inequality in~\cite[Theorem~3.5]{pinelis1994optimum}.

\begin{theorem}
    \label{th:finite_sample_guarantees_martingale_concentration}
    Let $\gamma \in (0,2\theta)$ and fix any
    $\delta\in(0,1)$. Then, with probability at least $1-\delta$ over the random draw $\mathcal D\sim \mathbb P^N$, the \gls*{SAA}
    error at the $x$-component of
    $\omega_K=\mathcal A_{\mathrm{DY}}(\mathcal D)$ satisfies %
    \begin{equation}
    \label{eq:saa_error_martingale_concentration_bound}
        \| \Delta_{\mathcal D}(\hat x_K) \| \leq
        \frac{2\eta_K}{\theta}
        +
        \nu_K
        \sqrt{
            2N\log\frac{2}{\delta}
        } =: \varphi(N, K, \delta)\,.
    \end{equation}
    Consequently, the true \gls*{DY} fixed-point residual satisfies
    \begin{equation}
    \label{eq:finite_sample_guarantees_martingale_concentration}
        \rho_{\mathrm{DY}}(\hat z_K)
        \leq
        \hat\rho_{\mathcal D}(\hat z_K)
        +
        \gamma
        \varphi(N, K, \delta).
    \end{equation}
\end{theorem}
\begin{proof}
    We begin by observing that, by the triangle inequality, the \gls*{SAA} error term $\|\Delta_{\mathcal{D}}(\hat{x}_K)\|$ can be bounded by
    \begin{equation}
    \label{eq:saa_error_decomposition_bias_martingale_fluctuation}
        \|\Delta_{\mathcal{D}}(\hat{x}_K)\| = \| M_N \| \leq  \| M_0 \| + \|M_N - M_0\|\,.
    \end{equation}
    The first term in the decomposition above is bounded by~\eqref{eq:saa_error_bias_control} in Lemma~\ref{le:saa_error_bias_control}. To control the second term in~\eqref{eq:saa_error_decomposition_bias_martingale_fluctuation}, we verify that the martingale concentration inequality of \cite[Theorem~3.5]{pinelis1994optimum} applies. To this end, we construct the sequence $\{f_0, f_1, \dots\}$, where $f_i = M_i - M_0$ for $i \in \{0, \dots, N\}$ and $f_i = f_N$ for all $i \in \N$ such that $i > N$. Then, $f_0 = 0$ by construction and $\{f_0, f_1, \dots\}$ is a martingale because so is $\{M_0, \dots, M_N\}$. Hence, since $f_i \in \R^n$ and every separable Hilbert space is also $(2,1)$-smooth~\cite{pinelis1994optimum}, it only remains to verify that $\sum_{i = 1}^{\infty} \| f_i - f_{i-1} \|_{\infty}^2$ is finite, where $\| f_i - f_{i-1} \|_{\infty}$ denotes the essential supremum of $\| f_i - f_{i-1} \|$ over all outcomes of $f_i - f_{i-1}$. By Lemma~\ref{le:doob_martingale_increment_bounded}, this condition holds since
    \begin{equation*}
        \sum_{i = 1}^{\infty} \| f_i - f_{i-1} \|_{\infty}^2 = \sum_{i = 1}^N \| M_i - M_{i-1} \|_{\infty}^2 \leq N \nu_K^2\,.
    \end{equation*} 
    Therefore, by applying \cite[Theorem~3.5]{pinelis1994optimum} with $D = 1$ and $b_*^2=N \nu_K^2$, we have that, for every $r\geq0$,
    \begin{equation*}
        \mathbb P\left(\sup_{i \geq 0} \|f_i\| \geq r\right)
        \leq 2 \exp \left(-\frac{r^2}{2N \nu_K^2} \right)\,,
    \end{equation*}
    from which we obtain that, with probability at least $1-\delta$,
    \begin{equation}
    \label{eq:saa_error_martingale_concentration_bound_step_1}
        \| M_N - M_0 \| \leq \sup_{i \geq 0} \|f_i\| \leq \nu_K
        \sqrt{2N\log\frac{2}{\delta}}\,.
    \end{equation}
    Finally,~\eqref{eq:saa_error_martingale_concentration_bound} follows by combining~\eqref{eq:saa_error_martingale_concentration_bound_step_1} with~\eqref{eq:saa_error_bias_control} and~\eqref{eq:saa_error_decomposition_bias_martingale_fluctuation}
\end{proof}

Differently from the upper bound on $\| \Delta_{\mathcal D}(\hat x_K) \|$ given by Theorem~\ref{th:finite_sample_guarantees_algorithmic_stability} in~\eqref{eq:finite_sample_guarantees_algorithmic_stability}, which depends on a specific dataset realization $\mathcal{D} \sim \mathbb{P}^N$ through the quantity $\hat{\ell}_{\mathcal{D}}(\omega_K)$, the bound~\eqref{eq:saa_error_martingale_concentration_bound} only involves quantities, such as problem parameters, that are known a priori, before the execution of the data-driven \gls*{FPI}~\eqref{eq:data_driven_davis_yin_splitting}. Most importantly, as $\eta_K \leq \frac{\gamma\varsigma K}{N}$ by Proposition~\ref{prop:one_sample_replacement_sensitivity} and recalling the definition of $\nu_K$ in~\eqref{eq:doob_martingale_increment_bounded}, Theorem~\ref{th:finite_sample_guarantees_martingale_concentration} yields the bound
\begin{equation}
\label{eq:saa_error_martingale_concentration_bound_explicit}
    \| \Delta_{\mathcal D}(\hat x_K) \| \leq \frac{2\gamma \varsigma K}{N\theta} + \varsigma \left(1 + \frac{2\gamma  K}{\theta}\right)\sqrt{\frac{2\log(2/\delta)}{N}}\,,
\end{equation}
which, unlike~\eqref{eq:finite_sample_guarantees_algorithmic_stability}, vanishes whenever $K = o(\sqrt N)$. Furthermore, similar to~\eqref{eq:epsilon_value_for_davis_yin_fixed_point_residual_bound_algorithimic_stability}, the generalization bound~\eqref{eq:finite_sample_guarantees_martingale_concentration} shows that~\eqref{eq:davis_yin_fixed_point_residual_bound} holds with arbitrarily high confidence $1-\delta$ and irrespective of $\mathbb{P}$ if the radius $\epsilon$ is such that
\begin{equation}
\label{eq:epsilon_value_for_davis_yin_fixed_point_residual_bound_martingale_inequalities}
    \epsilon \geq \frac{\hat{\rho}_{\mathcal{D}}(\hat{z}_K)}{\gamma} + \varphi(N, K, \delta) =: \epsilon_{\mathrm{MI}}(N, K, \delta)\,.
\end{equation}
In particular, we remark that~\eqref{eq:epsilon_value_for_davis_yin_fixed_point_residual_bound_algorithimic_stability} and~\eqref{eq:epsilon_value_for_davis_yin_fixed_point_residual_bound_martingale_inequalities} imply that the fixed-point residual bound~\eqref{eq:davis_yin_fixed_point_residual_bound} holds with confidence $1-\delta$ whenever $\epsilon \geq \operatorname{min}(\epsilon_{\mathrm{AS}}(N, K, \frac{\delta}{2}), \epsilon_{\mathrm{MI}}(N, K, \frac{\delta}{2}))$ by the union bound.

\subsection{Improved bounds under linear convergence}
The finite-sample certificates of Theorem~\ref{th:finite_sample_guarantees_algorithmic_stability} and Theorem~\ref{th:finite_sample_guarantees_martingale_concentration} crucially depend on the sensitivity $\eta_K$ of the $x$-component of the learned hypothesis $\omega_K = \mathcal{A}_{\mathrm{DY}}(\mathcal{D})$ to changes in the input dataset $\mathcal{D}$. Under maximal monotonicity of the operators $A$ and $B$ in~\eqref{eq:monotone_inclusion} and Assumption~\ref{ass:oracle_theta_cocoercive}, Proposition~\ref{prop:one_sample_replacement_sensitivity} yields the upper bound $\eta_K \leq \frac{\gamma \varsigma K}{N}$, which increases linearly with the iteration budget $K$. In this subsection, we show that this dependence on $K$ can be improved under standard assumptions guaranteeing contractivity of the \gls*{DY} fixed-point operator $T_{\mathrm{DY}}$ defined by~\eqref{eq:davis_yin_splitting} towards the fixed-point set \cite{davis2017three}.

To obtain a sharper bound on the replace-one sensitivity $\eta_K$, we require that at least one of the operators $A$, $B$, or $C$ in~\eqref{eq:monotone_inclusion} is strongly monotone, and at least one of $A$ or $B$ is Lipschitz continuous, as formalized in the following assumption~\cite{davis2017three}.
\begin{assumption}
\label{ass:davis_yin_linear_convergence}
     Let $\mu_A, \mu_B, \mu_C \geq 0$ be the strong monotonicity constants of $A$, $B$, and $C$, respectively. Moreover, let $L_A,L_B\in(0,+\infty]$ denote the Lipschitz constants of $A$ and $B$, with the convention that $1/L_A=0$ if $A$ is not Lipschitz and $1/L_B=0$ if $B$ is not Lipschitz. Then, it holds that
     \begin{equation*}
         (\mu_A + \mu_B + \mu_C) (1 /L_A + 1 / L_B) > 0\,.\footnote{As discussed in Section~\ref{sec:problem_formulation} and similarly to~\cite{hardt2016train, farnia2021train, fabiani2026finite}, in what follows we tacitly assume that the oracle $O(\cdot, \xi)$ satisfies the same properties postulated for the true operator $C$, uniformly for all $\xi \in \Xi$. Similarly, we assume that the solution set of~\eqref{eq:monotone_inclusion}, with $\widehat{C}_{\mathcal{D}}$ in place of $C$, is non-empty for all $\mathcal{D} \sim \mathbb{P}^N$.}
     \end{equation*}
\end{assumption}
Since averaging preserves strong monotonicity, under Assumption~\ref{ass:davis_yin_linear_convergence}, \cite[Theorem~D.6]{davis2017three} ensures that, for any step-size $\gamma \in (0, 2\beta)$, any initial point $\hat{z}_0 \in \R^n$, and any training dataset realization $\mathcal{D} \in \Xi^N$, there exists a constant $\alpha \in (0,1)$, given explicitly in \cite[Theorem~D.6]{davis2017three} and dependent of the specific scenario in Assumption~\ref{ass:davis_yin_linear_convergence} but not on $\mathcal{D}$, such that the sequence $\{\hat{z}_t\}_{t \in \N}$ defined by~\eqref{eq:data_driven_davis_yin_splitting} satisfies
\begin{equation}
\label{eq:davis_yin_contractivity_fixed_point_set}
    \| \hat{z}_{K} - z^\star_{\mathcal{D}}\| \leq \alpha^K \|\hat{z}_0 - z^\star_{\mathcal{D}} \|\,,
\end{equation}
where $z^\star_{\mathcal{D}}$ denotes a fixed-point of $\widehat{T}_{\mathcal{D}}$. Our next proposition leverages the exponential convergence bound of each trajectory of~\eqref{eq:data_driven_davis_yin_splitting} to its corresponding fixed-point set in~\eqref{eq:davis_yin_contractivity_fixed_point_set} to refine the estimate of the replace-one sensitivity $\eta_K$ given by~\eqref{eq:one_sample_replacement_sensitivity}.

\begin{proposition}
\label{prop:improved_replace_one_sensitivity_linear_convergence}
    Fix any $\gamma \in (0, 2\theta)$. Then, for any pair of neighboring datasets $\mathcal{D} \in \Xi^N$ and $\mathcal{D}^i$, with $i \in \{1, \dots, N\}$, the respective hypotheses $\omega_K = \mathcal{A}_{\mathrm{DY}}(\mathcal{D})$ and $\omega_K^{i} = \mathcal{A}_{\mathrm{DY}}(\mathcal{D}^{i})$ generated through the \gls*{FPI}~\eqref{eq:data_driven_davis_yin_splitting}, initialized with the same starting point $\hat{z}_0 = \hat{z}_0^i \in \R^n$, satisfy
    \begin{equation}
    \label{eq:improved_one_sample_replacement_sensitivity}
        \|\hat{x}_K - \hat{x}_K^{i}\| \leq 2 d_0 \alpha^K + \frac{\varsigma}{\bar{\mu} N}\,,
    \end{equation}
    where  $d_0 = \sup_{\mathcal{D} \in \Xi^N} \operatorname{dist}(\hat{z}_0, \operatorname{fix}(\widehat{T}_{\mathcal{D}}))$ is the maximum initial distance to the set of fixed-points of $\widehat{T}_{\mathcal{D}}$ and $\bar{\mu} = \mu_A + \mu_B + \mu_C$.
\end{proposition}
\begin{proof}
    For any pair of neighboring datasets $\mathcal{D} \in \Xi^N$ and $\mathcal{D}^i$, let $x^\star_\mathcal{D}$ and $x^\star_{\mathcal{D}^i}$ denote the solution to the monotone inclusion problem~\eqref{eq:monotone_inclusion} with $\widehat C_{\mathcal D}$ and $\widehat C_{\mathcal D^i}$ in place of $C$, respectively. In particular, note that $x^\star_\mathcal{D}$ and $x^\star_{\mathcal{D}^i}$ exist and are unique since the solution set of~\eqref{eq:monotone_inclusion} is non-empty by assumption and since, for all $\mathcal{D} \in \Xi^N$, $A+B+\widehat C_{\mathcal D} =: \widehat F_\mathcal{D}$ is $\bar{\mu}$-strongly monotone. We first observe that, by the triangle inequality, the left-hand side of~\eqref{eq:improved_one_sample_replacement_sensitivity} is upper bounded by
    \begin{equation}
    \label{eq:improved_one_sample_replacement_sensitivity_step_1}
        \|\hat{x}_K - \hat{x}_K^{i}\| \leq
        \|\hat x_K - x_{\mathcal D}^{\star}\| + \|\hat x_K^{i} - x_{\mathcal D^i}^{\star}\| + \|x_{\mathcal D}^{\star}-x_{\mathcal D^i}^{\star}\|\,.
    \end{equation}
    By nonexpansiveness of $J_{\gamma B}$ and using~\eqref{eq:davis_yin_contractivity_fixed_point_set}, both the first and second terms in the decomposition above can be bounded by
    \begin{equation}
    \label{eq:improved_one_sample_replacement_sensitivity_step_2}
        \|\hat x_K - x_{\mathcal D}^{\star}\| \leq \|\hat z_K - z_{\mathcal D}^{\star}\| \leq \alpha^K \|\hat{z}_0 - z^\star_{\mathcal{D}} \| \leq d_0 \alpha^K\,.
    \end{equation}
    It therefore remains to bound the third term in~\eqref{eq:improved_one_sample_replacement_sensitivity_step_1}. To this end, we observe that $0 \in \widehat F_{\mathcal{D}^i}(x^\star_{\mathcal{D}^i})$ implies 
    \begin{equation*}
        \widehat C_{\mathcal D}(x_{\mathcal D^i}^{\star}) - \widehat C_{\mathcal D^i}(x_{\mathcal D^i}^{\star}) \in \widehat{F}_{\mathcal D}(x_{\mathcal D^i}^{\star})\,.
    \end{equation*}
    Moreover, by $\bar{\mu}$-strong monotonicity of $F_{\mathcal D}$, since $0 \in \widehat F_{\mathcal{D}}(x^\star_{\mathcal{D}})$ by definition, and using Cauchy-Schwarz inequality, we have
    \begin{align}
        \nonumber
        &\bar\mu
        \|x_{\mathcal D}^{\star}-x_{\mathcal D^i}^{\star}\|^2
        \leq
        \langle \widehat C_{\mathcal D^i}(x_{\mathcal D^i}^{\star}) - \widehat C_{\mathcal D}(x_{\mathcal D^i}^{\star}), x_{\mathcal D}^{\star}-x_{\mathcal D^i}^{\star} \rangle \\
        \label{eq:improved_one_sample_replacement_sensitivity_step_3}
        &\qquad\leq
        \| \widehat C_{\mathcal D}(x_{\mathcal D^i}^{\star}) - \widehat C_{\mathcal D^i}(x_{\mathcal D^i}^{\star}) \|
        \|x_{\mathcal D}^{\star}-x_{\mathcal D^i}^{\star}\|\,.
    \end{align}
    Since $\mathcal D$ and $\mathcal D^i$ differ in one sample only, using Assumption~\ref{ass:oracle_replacement-sensitivity} and following~\eqref{eq:one_sample_replacement_sensitivity_step_2}, we obtain $\| \widehat C_{\mathcal D}(x_{\mathcal D^i}^{\star}) - \widehat C_{\mathcal D^i}(x_{\mathcal D^i}^{\star}) \| \leq \frac{\varsigma}{N}$, which, together with~\eqref{eq:improved_one_sample_replacement_sensitivity_step_3}, yields $\|x_{\mathcal D}^{\star}-x_{\mathcal D^i}^{\star}\| \leq \frac{\varsigma}{\bar\mu N}$, which concludes the proof by inspection of~\eqref{eq:improved_one_sample_replacement_sensitivity_step_1} and~\eqref{eq:improved_one_sample_replacement_sensitivity_step_2}.
\end{proof}

Proposition~\ref{prop:improved_replace_one_sensitivity_linear_convergence}
shows that, whenever Assumption~\ref{ass:davis_yin_linear_convergence} holds and the \gls*{DY} algorithm converges linearly, the dependence of the replace-one sensitivity $\eta_K$ on the iteration count improves from the scaling $\mathcal{O}\left(\frac{K}{N}\right)$ in Proposition~\ref{prop:one_sample_replacement_sensitivity} to $\mathcal{O}\left(\alpha^K + \frac{1}{N}\right)$, with $\alpha \in (0,1)$, showing that the sensitivity bound $\eta_K$ no longer deteriorates as the iteration budget increases. In particular, we note that substituting~\eqref{eq:improved_one_sample_replacement_sensitivity} in~\eqref{eq:generalization_correction_term_algorithmic_stability} and~\eqref{eq:saa_error_martingale_concentration_bound} yields corresponding improved finite-sample guarantees for the fixed-point residual~\eqref{eq:davis_yin_fixed_point_residual_bound}. Finally, we remark that, unlike the case of \gls*{FB} splitting methods considered in~\cite{fabiani2026finite}, Assumption~\ref{ass:davis_yin_linear_convergence} does not ensure contractivity of the \gls*{DY} fixed-point operator $T_{\mathrm{DY}}$ in~\eqref{eq:davis_yin_splitting}; for this reason, our bound~\eqref{eq:improved_one_sample_replacement_sensitivity} retains an exponentially decaying dependence on the iteration count $K$.

\section{Numerical experiments}
In this section, we present numerical simulations to validate the finite-sample guarantees of Theorem~\ref{th:finite_sample_guarantees_algorithmic_stability} and Theorem~\ref{th:finite_sample_guarantees_martingale_concentration}.\footnote{The source code that reproduces our numerical examples is available at \href{https://github.com/andrea-martin/data-driven-davis-yin.git}{github.com/andrea-martin/data-driven-davis-yin.git}.} For our experiments, we consider a portfolio optimization problem where the objective is to distribute the available resources in $n = 10$ different assets to minimize the investment risk while guaranteeing an expected return on the investments of at least $r \geq 0$. 
Formally, we model the distribution of our assets with a vector $x \in \mathcal{S}^n$, where the $i$-th component of $x$ represents the percentage of capital invested in $i$-th asset and 
\begin{equation*}
    \mathcal{S}^n = \{x \in \R^n: \boldsymbol{1}^\top x = 1, ~x \geq 0 \}\,,
\end{equation*}
denotes the standard simplex. We model the asset returns over the considered investment period as $R = m + \xi$, where $m \in \R^n$ collects the expected return of each asset and $\xi \sim \mathbb{P}$ is a zero-mean stochastic deviation, and measure the portfolio return corresponding to an allocation strategy $x \in \mathcal{S}^n$ by $(m + \xi)^\top x$. Last, following~\cite{davis2017three}, we augment the standard portfolio risk, measured by the variance of asset returns $x^\top \Sigma x$, with a term encouraging diversity of investments among assets. With these modeling choices, the portfolio optimization problem reads as
\begin{alignat}{3}
    \label{eq:portfolio_optimization_problem}
    &~\min_{x \in \mathcal{S}^n} ~ \mathbb{E}_{\xi \sim \mathbb{P}}[(\xi^\top x)^2] + \psi \|x\|^2 = &&~\min_{x \in \mathcal{S}^n} ~ x^\top (\Sigma + \psi I) x\\
    \nonumber
    &\st~ m^\top x \geq r &&\st ~ m^\top x \geq r\,,
    \nonumber
\end{alignat}
where $\psi > 0$ is a regularization parameter and $\Sigma \succeq 0$ is the covariance matrix of $\xi \sim \mathbb{P}$. Let $\mathcal{C} = \{x \in \R^n : \mathbb{E}_{\xi \sim \mathbb{P}}[R^\top x] = m^\top x \geq r\}$ denote the set of strategies satisfying the minimum expected return constraint. Then, the portfolio optimization problem~\eqref{eq:portfolio_optimization_problem} can be equivalently rewritten as: 
\begin{equation*}
    \min_{x \in \R^n} ~x^\top (\Sigma + \psi I) x + i_{\mathcal{C}}(x) + i_{\mathcal{S}^n}(x)\,,
\end{equation*}
where $i_{\mathcal{X}}$ is the indicator function of the set $\mathcal{X}$. Since $\mathcal{C}$ and $\mathcal{S}^n$ are both closed convex polyhedral sets, the first-order optimality condition of~\eqref{eq:portfolio_optimization_problem} yields
an instance of~\eqref{eq:monotone_inclusion} with $A = \operatorname{N}_{\mathcal{C}} = \partial i_\mathcal{C}$, $B = \operatorname{N}_{\mathcal{S}^n} = \partial i_{\mathcal{S}^n}$, and $C(x) = 2(\Sigma + \psi I)x$, because the subdifferential of any proper lower semicontinuous convex function is maximally monotone by~\cite[Theorem~20.25]{bauschke2017convex} and the map $C(x)$ is $\theta$-cocoercive, with $\theta = \frac{1}{2(\lambda_{\max}(\Sigma) + \psi)}$, by the Baillon–Haddad theorem~\cite[Corollary~18.17]{bauschke2017convex}.

As the resolvents of a normal cone operator correspond to Euclidean projections, applying the \gls*{DY} splitting~\eqref{eq:davis_yin_splitting} to the portfolio optimization problem~\eqref{eq:portfolio_optimization_problem} yields the recursion
\begin{subequations}
    \label{eq:davis_yin_splitting_portfolio_optimization}
    \begin{align}
        \label{eq:davis_yin_splitting_x_portfolio_optimization}
        x_t^{[i]} &= (z_t^{[i]} - \tau)_+\,, ~\forall i \in \{1, \dots, n\}\,,\\
        \label{eq:davis_yin_splitting_v_portfolio_optimization}
        v_t &= 2 x_t - z_t - 2\gamma(\Sigma + \psi I)x_t\,,\\
        \label{eq:davis_yin_splitting_y_portfolio_optimization}
        y_t &= v_t + \frac{(r - m^\top v_t)_+}{\| m \|^2}m\,,\\
        \label{eq:davis_yin_splitting_z_portfolio_optimization}
        z_{t+1} &= z_t + y_t - x_t\,,
    \end{align}
\end{subequations}
where $x_t^{[i]}$ denotes the $i$-th component of the vector $x_t$, $(a)_+ = \max\{a, 0\}$ represents the positive part of $a \in \R$, and $\tau$ in~\eqref{eq:davis_yin_splitting_x_portfolio_optimization} is given by $\tau = \frac{1}{\rho}\left(\sum_{i = 1}^\rho u_t^{[i]} - 1\right)$, where
\begin{equation*}
    \rho = \max\left\{j \in \{1, \dots, n\} : u_t^{[j]} - \frac{1}{j}\left(\sum_{i= 1}^j u_t^{[i]} - 1\right) > 0\right\}\,,
\end{equation*}
and $u_t$ is a permutation of $z_t$ such that $u_t^{[1]} \geq \dots \geq u_t^{[n]}$; see, for instance,~\cite{duchi2008efficient} for a complete derivation. In particular, we observe that, while projecting onto $\mathcal{C} \cap \mathcal{S}^n$ generally requires solving a quadratic program, all algorithmic steps in~\eqref{eq:davis_yin_splitting_portfolio_optimization} can be computed in nearly closed-form, making the \gls*{DY} splitting~\eqref{eq:davis_yin_splitting} computationally more efficient than \gls*{FB} methods, especially for large-scale instances of~\eqref{eq:portfolio_optimization_problem}. 

As the probability distribution $\mathbb{P}$ underlying the asset returns is generally unknown, we then replace the forward evaluation of $C$ in~\eqref{eq:davis_yin_splitting_v_portfolio_optimization} with the associated \gls*{SAA}~\eqref{eq:data_driven_approximation_C_saa}. To this end, we rely on the stochastic oracle $O$ defined by
\begin{equation}
\label{eq:oracle_portfolio_optimization}
    O(x, \xi) = 2( \xi \xi^T + \psi I ) x\,,
\end{equation}
and a finite set of training samples $\{\xi^1, \dots, \xi^N\}$ generated as follows. 
Based on standard factor covariance models~\cite{ross2013arbitrage}, we first randomly select a dense matrix $\Sigma \succeq 0$ to capture correlations between assets, and normalize it so that $\operatorname{tr}(\Sigma) = 1$. Then, we let $\xi = U \Lambda^{\frac{1}{2}}s$, where $\Sigma = U \Lambda U^\top$ is the singular value decomposition of $\Sigma$ and each component of $s$ is an independent Rademacher random variable. In particular, this choice ensures that the oracle $O$ in~\eqref{eq:oracle_portfolio_optimization} is unbiased, is $\hat{\theta}$-cocoercive, with $\hat{\theta} = \frac{1}{2(1+\psi)}$, since $\| \xi \| = 1$ by construction, and that it satisfies Assumption~\ref{ass:oracle_replacement-sensitivity} with $\varsigma = 2$, since
\begin{align*}
    \|O(x, \xi) - O(x, \xi^\prime)\| &= 2 \| \xi \xi^\top x - \xi^\prime {\xi^\prime}^\top x\| \\
    &\leq 2 \| \xi \xi^\top - \xi^\prime {\xi^\prime}^\top\|\| x\| \\
    &= 2\sqrt{1 - (\xi^\top \xi^\prime)^2} \| x \| \leq 2\|x\|_1 \leq 2\,,
\end{align*}
for all $x \in \mathcal{S}^n$.\footnote{For this example, it suffices to verify Assumption~\ref{ass:oracle_replacement-sensitivity} for all $x \in \mathcal{S}^n$ since the resolvent $J_{\gamma B}$ is the projection operator onto the standard simplex.} Based on~\eqref{eq:oracle_portfolio_optimization}, we define the data-driven counterpart of the \gls*{FPI}~\eqref{eq:davis_yin_splitting_portfolio_optimization} by replacing the update~\eqref{eq:davis_yin_splitting_v_portfolio_optimization} by
\begin{equation}
\label{eq:data_driven_davis_yin_splitting_v_portfolio_optimization}
    \hat{v}_t = 2\hat{x}_t - \hat{z}_t - 2\gamma(\widehat{\Sigma}_{\mathcal{D}} + \psi I)\hat{x}_t\,, \quad \widehat{\Sigma}_{\mathcal{D}} = \frac{1}{N} \sum_{i = 1}^N \xi^{i} {\xi^i}^\top\,.
\end{equation}

To validate the finite-sample guarantees of Theorem~\ref{th:finite_sample_guarantees_algorithmic_stability} and Theorem~\ref{th:finite_sample_guarantees_martingale_concentration}, we then unroll the \gls*{FPI}~\eqref{eq:davis_yin_splitting_portfolio_optimization}, with~\eqref{eq:data_driven_davis_yin_splitting_v_portfolio_optimization} in place of~\eqref{eq:davis_yin_splitting_v_portfolio_optimization}, for $K= 50$ iterations, using different sample sizes $N \in [10^1, 10^{18}]$. Figure~\ref{fig:portfolio_optimization_K50} collects the results obtained with starting point $z_0 = \frac{1}{10} \mathbf{1}$, regularization parameter $\psi = 10^{-3}$, confidence parameter $\delta = 10^{-6}$, step size $\gamma = 0.899 \in (0, 2\hat{\theta})$, an expected return $m^{[i]}\sim \mathcal{U}_{[0.04, 0.14]}$ for all assets, and target return $r = 0.11$, chosen such that $r < \max_{i \in \{1, \dots, n\}} ~m^{[i]}$ to ensure feasibility after drawing $m$. 
\begin{figure}[htb]
\centering
\includegraphics[width=\columnwidth]{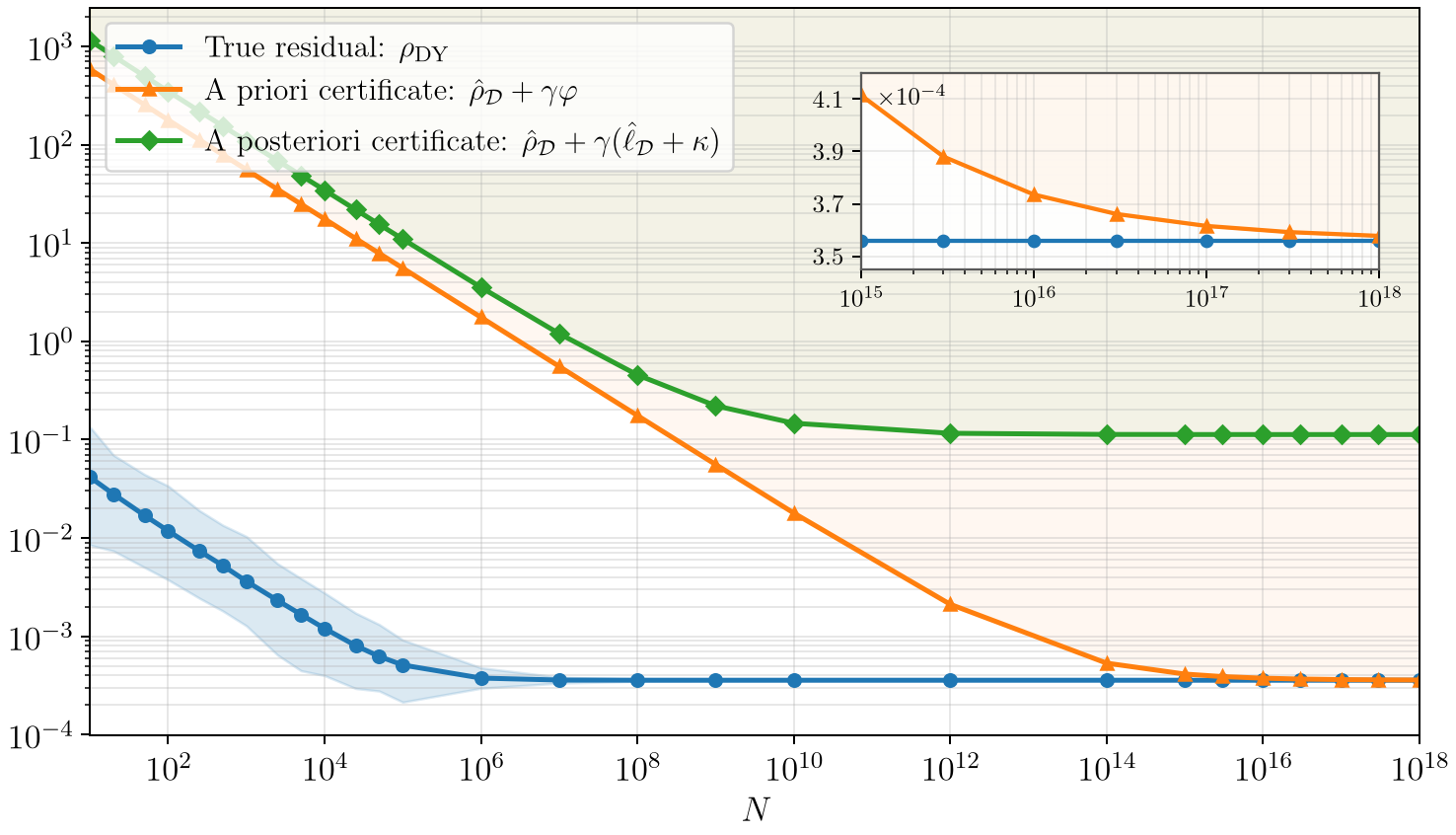}
\caption{Evolution of the fixed-point residual $\rho_{\mathrm{DY}}(\hat{z}_{50})$ associated with the hypothesis $\omega_{50} = \mathcal{A}_{\mathrm{DY}}(\mathcal{D})$, as a function of the number of samples in $\mathcal{D}$. The shaded blue area represents the smallest envelope collecting all the realizations of $\rho_{\mathrm{DY}}(\hat{z}_{50})$ corresponding to $1000$ random draws of $\mathcal{D}$.
}
\label{fig:portfolio_optimization_K50}
\end{figure}

We observe that the true fixed-point residual $\rho_{\mathrm{DY}}(\hat{z}_{50})$, shown in blue in~Figure~\ref{fig:portfolio_optimization_K50}, initially decreases as the sample size increases, before flattening at around $3.6 \cdot 10^{-4}$. This is because, while adding more samples when $N$ is small is beneficial as it improves the quality of the \gls*{SAA}~\eqref{eq:data_driven_davis_yin_splitting_v_portfolio_optimization}, for sufficiently large $N$, the difference between~\eqref{eq:data_driven_davis_yin_splitting_v_portfolio_optimization} and~\eqref{eq:davis_yin_splitting_v_portfolio_optimization} becomes negligible, and further reducing the population residual requires increasing the iteration budget $K$. Furthermore, we observe that the correction term~\eqref{eq:saa_error_martingale_concentration_bound} of the a priori finite-sample certificate of Theorem~\ref{th:finite_sample_guarantees_martingale_concentration} converges to zero, yielding asymptotically tight bounds on the true fixed-point residual $\rho_{\mathrm{DY}}(\hat{z}_{50})$ in terms of its empirical counterpart $\hat{\rho}_{\mathcal{D}}(\hat{z}_{50})$. Conversely, consistently with the discussion in Example~\ref{ex:average_oracle_discrepancy_does_not_vanish}, we note the a posteriori certificate~\eqref{eq:finite_sample_guarantees_algorithmic_stability} plateaus at approximately $0.112$ due to the presence of the non-vanishing term $\hat{\ell}_{\mathcal{D}}$.

\section{Conclusion}
In this paper, we analyzed a data-driven \gls*{DY} splitting algorithm for solving structured monotone inclusion problems in which one constituent operator is unavailable and its forward evaluation is replaced by a \gls*{SAA} constructed from finitely many noisy oracle samples. We showed that certifying the true fixed-point residual amounts to controlling the \gls*{SAA} error at the output of the data-driven \gls*{DY} algorithm, and we highlighted how bounding this error by the expectation of a surrogate loss function, as previously proposed in the literature to apply algorithmic stability arguments, can lead to conservative bounds whose statistical excess need not shrink with the sample size. We then observed that the conditional expectations of the \gls*{SAA} error term along the filtration generated by the samples form a Doob martingale, and we established asymptotically consistent a priori residual certificates via martingale inequalities under different monotonicity assumptions on the operators defining the inclusion problem. Future work includes extending our analysis to other data-driven operator splitting algorithms, relaxing the \gls*{iid} assumption on the samples in the dataset, and considering stochastic oracle approximations in which new samples are drawn at each iteration of the algorithm.

\bibliographystyle{IEEEtran}
\bibliography{references}

\end{document}